\newif\ifDraft 
\newif\ifAnon 
\newif\ifFull 
\newif\ifIEEE 
\newif\ifICLR 
\newif\ifICML 
\newif\ifNIPS 
\Draftfalse 
\Anonfalse
\Fulltrue
\IEEEfalse
\ICMLfalse
\NIPSfalse
 
\ifICLR
	\documentclass{article} 
	\usepackage{iclr2026_conference,times}
\fi
\ifICML
	\documentclass{article}
	\usepackage{icml2026}
\fi
\ifFull
    \documentclass{article} 
\fi
\ifNIPS
	\documentclass{article}
	\PassOptionsToPackage{numbers}{natbib}
	\usepackage{neurips_2026}
\fi
 
\usepackage[english]{babel}
\usepackage[utf8]{inputenc}
\usepackage[T1]{fontenc}

\usepackage{etoolbox} 
\usepackage{xparse}

\usepackage{lmodern}

\usepackage[table]{xcolor}
\usepackage{booktabs,makecell}
\usepackage{supertabular}
\usepackage{multirow}
\usepackage{longtable}

\usepackage{amssymb}
\usepackage{amsmath}
\usepackage{amsfonts}
\usepackage{amsthm}
\usepackage{nccmath}
\usepackage{bbm}
\usepackage{bm} 
\usepackage[b]{esvect}
\usepackage{mathrsfs}
\usepackage{mathtools}
\usepackage{centernot}
\usepackage{xspace}
\usepackage{xifthen}
\usepackage{stmaryrd}
\usepackage{pifont}
\usepackage{underscore} 

\usepackage{totcount}
\usepackage{color}
\usepackage[colorinlistoftodos]{todonotes}
\usepackage{verbatim}

\usepackage{graphicx}
\usepackage{caption}
\usepackage{subcaption}

\usepackage{lscape}
\usepackage{float}
\usepackage[shortlabels,inline]{enumitem}
\usepackage{varwidth} 
\usepackage[most]{tcolorbox} 
\usepackage[export]{adjustbox}
\usepackage{authblk}
\usetikzlibrary{positioning}

\definecolor{linkblue}{HTML}{00003c}
\usepackage[colorlinks,breaklinks=true,allcolors=linkblue]{hyperref} 
\usepackage[capitalize,noabbrev]{cleveref} 

\usepackage{algorithm}
\usepackage[noend]{algpseudocode}

\makeatletter 
\ifFull
	\renewcommand*\l@section[2]
	{%
		\ifnum \c@tocdepth >\z@
		\addpenalty \@secpenalty
		\addvspace {1.0em \@plus \p@ }%
		\setlength \@tempdima {1.5em}%
		\begingroup
		\parindent \z@
		\rightskip \@pnumwidth
		\parfillskip -\@pnumwidth
		\leavevmode \bfseries
		\advance \leftskip \@tempdima
		\hskip -\leftskip
		#1\nobreak
		\hfil
		\nobreak
		\hb@xt@ \@pnumwidth {\hss #2\kern -\p@ \kern \p@ }%
		\par
		\endgroup
		\fi
	}
\fi
\makeatother

\setlist{itemsep=2pt,topsep=2pt,parsep=2pt,partopsep=2pt} 

\definecolor{codegreen}{rgb}{0,0.6,0}
\definecolor{codegray}{rgb}{0.5,0.5,0.5}
\definecolor{codepurple}{rgb}{0.58,0,0.82}
\definecolor{backcolour}{rgb}{0.95,0.95,0.92}
\definecolor{mDarkBrown}{HTML}{604c38}
\definecolor{mDarkTeal}{HTML}{23373b}
\definecolor{mLightBrown}{HTML}{EB811B}
\definecolor{mLightGreen}{HTML}{14B03D}
\definecolor{lightgreen}{rgb}{0.86, 0.93, 0.78}
\definecolor{textgreen}{rgb}{0.0, 0.5, 0.0}
\definecolor{bordergreen}{rgb}{0.55, 0.76, 0.74}
\definecolor{lightblue}{rgb}{0.70, 0.90, 0.99}
\definecolor{borderblue}{rgb}{0.01, 0.66, 0.96}
\definecolor{lightamber}{rgb}{1, 0.93, 0.70}
\definecolor{borderamber}{rgb}{1, 0.76, 0.03}
\definecolor{lightcolor4}{rgb}{ 0.93, 0.70, 1}
\definecolor{bordercolor4}{rgb}{0.76, 0.03, 1}
\definecolor{lightcolor5}{rgb}{0.78,0.86,0.93}
\definecolor{bordercolor5}{rgb}{0.74,0.55,0.76}
\definecolor{blueinfo}{RGB}{64, 112, 173}
\definecolor{greeninfo}{RGB}{148, 176, 54}
\definecolor{yellowinfo}{RGB}{240, 189, 82}
\definecolor{redinfo}{RGB}{194, 77, 59}
\definecolor{devpurple}{RGB}{125, 58, 193}
\definecolor{implgray}{RGB}{193, 193, 193}
\definecolor{greencomment}{RGB}{118,191,105}
\definecolor{bluecomment}{RGB}{105,166,191}
\definecolor{yellowcomment}{RGB}{239,200,115}
\definecolor{redcomment}{RGB}{194, 77, 59}
\definecolor{yellowdark}{RGB}{156, 151, 61}
\definecolor{MyTeal}{RGB}{0,128,128}

\newcommand{\codecomment}[1]{{\footnotesize\textcolor{black!60!white}{/\!\!/~\!#1}}}
\renewcommand{\Comment}[1]{\codecomment{#1}}

\definecolor{ygreen}{RGB}{37, 133, 13}

\makeatletter
\algnewcommand{\ExtendedState}[1]{\State
	\parbox[t]{\dimexpr\linewidth-\ALG@thistlm}{\hangindent=\algorithmicindent\strut\hangafter=3#1\strut}}
\makeatother

\algnewcommand\algorithmicinput{\textbf{Input:}}
\algnewcommand\Input{\item[\algorithmicinput]}

\algrenewcommand{\algorithmiccomment}[1]{{\color{gray}// #1}}
\algnewcommand{\IIf}[1]{\State\algorithmicif\ #1\ \algorithmicthen}
\algnewcommand{\EndIIf}{\unskip\ \algorithmicend\ \algorithmicif}

\makeatletter
\newcommand*\ifcounter[1]{%
	\ifcsname c@#1\endcsname
	\expandafter\@firstoftwo
	\else
	\expandafter\@secondoftwo
	\fi
}
\makeatother

\NewDocumentEnvironment{algorithmWithNumbering}{m}{%
	\begin{algorithmic}[1]
		\ifcounter{ALG@line@#1}{}{\newcounter{ALG@line@#1}}
		\setcounter{ALG@line}{\value{ALG@line@#1}}
	}{%
		\setcounter{ALG@line@#1}{\value{ALG@line}}
	\end{algorithmic}%
}

\RequirePackage{color}
\RequirePackage[most]{tcolorbox}
\RequirePackage{totcount}
\tcbuselibrary{skins,breakable}

\IfFileExists{fontawesome.sty}{%
	\RequirePackage{fontawesome}
}{

} 

\newtcolorbox[use counter=alg,list inside=alg,crefname={Algorithm}{Algorithms}]{abox}[6][]{
	enhanced,
	colframe=black,
	colback=white,
	boxrule={#4},
	arc={#3},
	auto outer arc,
	pad at break*=0pt,
	vfill before first,
	before={\par\medskip\noindent},
	after={},
	top=12pt, left=4pt, enlarge top by=6pt,
	title={\rule[-.3\baselineskip]{0pt}{\baselineskip}\normalsize\sffamily\bfseries #2}, 
	varwidth boxed title*=-10pt, 
	attach boxed title to top left={yshift=-10pt,xshift=10pt}, 
	coltitle=black,
	boxed title style={colback=white,boxrule={#6},arc={#5},auto outer arc},
	#1
}

\makeatletter
\newcommand{\DrawLine}{%
	\begin{tikzpicture}
		\path[use as bounding box] (0,0) -- (\linewidth,0);
		\draw[color=black!75,dashed,dash phase=2pt]
		(0-\kvtcb@leftlower-\kvtcb@boxsep,0)--
		(\linewidth+\kvtcb@rightlower+\kvtcb@boxsep,0);
	\end{tikzpicture}%
}
\makeatother

\newtcolorbox{mybox}[2][]{%
	enhanced,
	title        = {#2},
	attach boxed title to top left={xshift=+3mm,yshift*=-3mm},
	breakable    = true,
	colback      = black!4,
	colframe     = black!75,
	fonttitle    = \bfseries,
	colbacktitle = black!10!white,
	coltitle     = black,
	#1
}

\newtcolorbox[use counter=func,list inside=func]{functionality}[2][]{%
	enhanced,
	title        = {Functionality~\thetcbcounter: #2},
	attach boxed title to top left={xshift=+3mm,yshift*=-3mm},
	breakable    = true,
	colback      = yellow!4,
	colframe     = black!75,
	fonttitle    = \bfseries,
	fontupper    = \small,
	fontlower    = \small,
	colbacktitle = yellow!10!white,
	coltitle     = black,
	#1
}

\newtcolorbox[use counter=prot,list inside=prot]{protocol}[2][]{%
	enhanced,
	title        = {Protocol~\thetcbcounter: #2},
	attach boxed title to top left={xshift=+3mm,yshift*=-3mm},
	breakable    = true,
	colback      = black!4,
	colframe     = black!75,
	fonttitle    = \bfseries, 
	fontupper    = \small,
	fontlower    = \small,
	colbacktitle = black!10!white,
	coltitle     = black,
	#1
}

\newtcolorbox[use counter=proc,list inside=proc]{procedure}[2][]{%
	enhanced,
	title        = {Procedure~\thetcbcounter: #2},
	attach boxed title to top left={xshift=+3mm,yshift*=-3mm},
	breakable    = true,
	colback      = cyan!2,
	colframe     = black!75,
	fonttitle    = \bfseries, 
	fontupper    = \small,
	fontlower    = \small,
	colbacktitle = cyan!5!white,
	coltitle     = black,
	#1
}

\newtcbox{\xmybox}[1][red]{on line,
	arc=7pt,colback=#1!10!white,colframe=#1!50!black,
	before upper={\rule[-3pt]{0pt}{10pt}},boxrule=1pt,
	boxsep=0pt,left=6pt,right=6pt,top=2pt,bottom=2pt}

\renewcommand{\vec}[1]{\bm{#1}}
\newcommand{\cC}{\mathcal{C}}
\newcommand{\cD}{\mathcal{D}}
\newcommand{\cE}{\mathcal{E}}

\newcommand{\cI}{\mathcal{I}}

\newcommand{\cL}{\mathcal{L}}
\newcommand{\cM}{\mathcal{M}}

\newcommand{\cP}{\mathcal{P}}

\newcommand{\cR}{\mathcal{R}}
\newcommand{\cS}{\mathcal{S}}
\newcommand{\cT}{\mathcal{T}}

\newcommand{\cV}{\mathcal{V}}

\newcommand{\cY}{\mathcal{Y}}

\newcommand{\vecf}{\ensuremath{\mathbf{f}}}

\newcommand{\vecj}{\ensuremath{\mathbf{j}}}

\newcommand{\vecl}{\ensuremath{\mathbf{l}}}

\newcommand{\vectt}{\ensuremath{\mathbf{t}}}
\newcommand{\vecu}{\ensuremath{\mathbf{u}}}
\newcommand{\vecv}{\ensuremath{\mathbf{v}}}
\newcommand{\vecw}{\ensuremath{\mathbf{w}}}
\newcommand{\vecx}{\ensuremath{\mathbf{x}}}
\newcommand{\vecy}{\ensuremath{\mathbf{y}}}
\newcommand{\vecz}{\ensuremath{\mathbf{z}}}
\newcommand{\veczero}{\ensuremath{\mathbf{0}}}

\newcommand{\xmath}[1]{\ensuremath{#1}\xspace}
\newcommand{\command}[1]{\textnormal{\textsc{#1}}} 
\newcommand{\cmd}[1]{\command{#1}} 

\renewcommand{\gets}{\leftarrow}
\newcommand{\rgets}[1]{%
\ifthenelse{\isempty{#1}}{\xleftarrow{\$}}{\xleftarrow{#1}}
}

\newcommand{\set}[1]{\left\{ { #1 }\right\}}

\let\emptyset\varnothing
\newcommand{\setpred}[2]{\xmath{\set{\begin{matrix}#1\end{matrix}\,:\ \begin{matrix}#2\end{matrix}}}}

\DeclarePairedDelimiter\floor{\lfloor}{\rfloor}

\newcommand{\RR}{\mathbb{R}}
\newcommand{\NN}{\mathbb{N}}

\newcommand{\FF}{\mathbb{F}}

\newcommand{\Func}[1]{\mathcal{F}_{\mathsf{#1}}}

\newcommand{\ZKP}{\mathsf{ZKP}}

\newcommand{\simulator}{\mathcal{S}}
\newcommand{\rel}{\xmath{\mathcal{R}}}

\newcommand{\statement}{\mathsf{x}}
\newcommand{\stm}{\statement}
\newcommand{\witness}{\mathsf{w}}
\newcommand{\wit}{\witness}

\newcommand{\com}{\mathsf{com}}
\newcommand{\commit}{\mathsf{Com}}

\newcommand{\hcom}[1]{\mathsf{Hcom}_{#1}}

\newcommand{\Fcp}{\hyperref[func:cp]{\ensuremath{\Func{CP}}}}
\newcommand{\auth}[1]{\left\llbracket #1 \right\rrbracket}

\newcommand{\binID}{\mathsf{binID}}
\newcommand{\edges}{\mathsf{edges}}
\newcommand{\gain}{\mathsf{gain}}
\newcommand{\treepath}{\mathsf{path}}

\newcommand{\buildtree}{\mathsf{BuildTree}}
\newcommand{\findsplit}{\mathsf{FindSplit}}

\newcommand{\Avail}{\mathsf{Avail}}
\newcommand{\padtoheight}{\mathsf{PadToHeight}}
\newcommand{\accgap}{\ensuremath{|\Delta|}}
\newcommand{\certxgb}{\xmath{\hyperref[proc:cert-xgboost]{\mathsf{CertXGB}}}}
\newcommand{\fdummy}{f_{\text{dum}}}
\newcommand{\bdummy}{b_{\text{dum}}}
\newcommand{\tdummy}{t_{\text{dum}}}

\newcommand{\prebinfp}{\xmath{\hyperref[prot:prebin-fp]{\mathsf{PreBinFP}}}}
\newcommand{\validatelogit}{\xmath{\hyperref[prot:validate-logit]{\mathsf{ValidateLogit}}}}
\newcommand{\logitfromprobfp}{\xmath{\hyperref[eq:logitfp]{\mathsf{LogitFromProbFP}}}}
\newcommand{\validateinference}{\xmath{\hyperref[prot:validate-inference]{\mathsf{ValidateInference}}}}
\newcommand{\inithists}{\xmath{\hyperref[prot:init-hists]{\mathsf{InitHists}}}}
\newcommand{\validateleafweights}{\xmath{\hyperref[prot:validate-leaf-weights]{\mathsf{ValidateLeafWeights}}}}
\newcommand{\validatesplits}{\xmath{\hyperref[prot:validate-splits]{\mathsf{ValidateSplits}}}}
\newcommand{\sigmoidwidefp}{\xmath{\hyperref[eq:sigmoidfp]{\mathsf{SigmoidWideFP}}}}
\newcommand{\clipfp}{\xmath{\hyperref[prot:ClipFP]{\mathsf{ClipFP}}}}

\newcommand{\fp}[1]{\tilde{#1}}
\newcommand{\gainfp}{\widetilde{\mathsf{gain}}}
\newcommand{\scale}{\mathsf{scale}}
\newcommand{\trainfp}{\xmath{\hyperref[proc:Train]{\mathsf{TrainXGB}}}}
\newcommand{\buildtreefp}{\mathsf{BuildTreeFP}}
\newcommand{\findsplitfp}{\mathsf{FindSplitFP}}

\newcommand{\certforest}{\xmath{\hyperref[proc:cert-forest]{\mathsf{CertForest}}}}
\newcommand{\inithistslabel}{\xmath{\hyperref[prot:init-hists-label]{\mathsf{InitHistsLabel}}}}
\newcommand{\validateleafweightslabel}{\xmath{\hyperref[prot:validate-leaf-weights-label]{\mathsf{ValidateLeafWeightsLabel}}}}
\newcommand{\validatesplitsgini}{\xmath{\hyperref[prot:validate-splits-gini]{\mathsf{ValidateSplitsGini}}}}

\ifDraft
	\usepackage[multiuser,inline,nomargin,draft]{fixme}
\else
	\usepackage[multiuser,inline,nomargin]{fixme}
\fi
\FXRegisterAuthor{akira}{eakira}{\color{redcomment}Akira}
\FXRegisterAuthor{antigoni}{eantigoni}{\color{devpurple}Antigoni}
\FXRegisterAuthor{nikolas}{enikolas}{\color{MyTeal}Nikolas}
\FXRegisterAuthor{chenkai}{echenkai}{\color{greencomment}Chenkai}
\FXRegisterAuthor{jiayi}{ejiayi}{\color{bluecomment}Jiayi}
\fxusetheme{color}

\newcommand{\akira}[1]{\akiranote{#1}}

\definecolor{emailgreen}{RGB}{11,102,35}
\newcommand{\email}[1]{\href{mailto:#1}{\textcolor{emailgreen}{\texttt{\nolinkurl{#1}}}}}

\newtheorem{claim}{Claim}

\ifFull
\usepackage[a4paper,margin=2.8cm]{geometry} 
\fi

\ifICML
\icmltitlerunning{ZKBoost: Zero-Knowledge Verifiable Training for XGBoost}
\fi

\ifFull
\usepackage{authblk}
\title{ZKBoost: Zero-Knowledge Verifiable Training for XGBoost}
\author[1]{Nikolas Melissaris}
\author[2]{Antigoni Polychroniadou} 
\author[2]{Akira Takahashi}
\author[3]{Chenkai Weng}
\author[3]{Jiayi Xu}
\affil[1]{CNRS, IRIF, Université Paris Cité}
\affil[2]{JPMorgan AI Research \& AlgoCRYPT CoE}
\affil[3]{Arizona State University}
\date{}
\fi

\ifNIPS
\title{ZKBoost: Zero-Knowledge Verifiable Training for XGBoost}
\author{}
\fi

\begin{document}

\ifFull
\begingroup
\renewcommand\thefootnote{\arabic{footnote}}

\footnotetext[1]{\email{nikolas@irif.fr}}
\footnotetext[2]{\email{{antigoni.polychroniadou, akira.takahashi}@jpmorgan.com}}
\footnotetext[3]{\email{{jiayixu7, chenkai.weng}@asu.edu}}
\endgroup
\fi

\ifFull
	\let\oldaddcontentsline\addcontentsline
	\def\addcontentsline#1#2#3{}
	\maketitle
	\def\addcontentsline#1#2#3{\oldaddcontentsline{#1}{#2}{#3}}
	
\fi


\listoffixmes 
 	\begin{abstract}
\akira{Draft Mode}
Gradient boosted decision trees, particularly XGBoost, are among the most effective methods for tabular data. As deployment in sensitive settings increases, cryptographic guarantees of model integrity become essential. We present ZKBoost, the first zero-knowledge proof of training (zkPoT) protocol for XGBoost, enabling model owners to prove correct training on a committed dataset without revealing data or model parameters.
Naively re-executing XGBoost training in ZK would incur prohibitive costs, primarily due to the oblivious partitioning of training samples and unknown tree splits. Moreover, previous work on ZKP of training and inference had subtle security issues, such as leakage of tree topology and soundness gaps allowing cheating model providers to deviate from the correct execution of training and inference. We make two key contributions to address these challenges: (1) a generic zkPoT template for XGBoost that can be instantiated with any general-purpose ZKP backend, significantly improving prover costs compared to naive re-execution of the training process; and (2) a VOLE-based instantiation that overcomes the security issues of previous ZK proofs of training at minimal costs. 
To maximize efficiency, we develop a fixed-point version of XGBoost, which is particularly well suited for efficient instantiation of ZKP, and show it matches standard XGBoost accuracy to within 1\% on real-world datasets.
\end{abstract}

    \ifFull
        \clearpage
		\tableofcontents 
		\clearpage
	\fi

    \section{Introduction}
\label{sec:intro}



As reliance on ML models grows, concerns about integrity and accountability are increasing.
Today, when a model $\cM$ is deployed, there is no evidence that it was truly obtained by training on dataset $\cD$ with prescribed hyperparameters. 
This gap is important. 
A dishonest provider could ship a hand-crafted model, mix in unauthorized data, or skip parts of training.
Clients have no way to distinguish such shortcuts from genuine training, because model providers in practice are not willing to reveal proprietary training data or the model parameters due to privacy and intellectual property concerns.
An emerging approach to address this gap is zero knowledge proofs (ZKP) \cite{DBLP:journals/siamcomp/GoldwasserMR89}, which allows a \emph{prover} to convince a \emph{verifier} of the truth of a statement without revealing anything beyond its validity. 
In the context of machine learning, it enables \emph{zero knowledge Proof of Training} (zkPoT), allowing a provider (prover) to convince clients (verifier) that ``private model $\cM$ is the result of a public training algorithm $\mathsf{Train}$ on cryptographically committed dataset $\cD$'' without exposing either $\cD$ or intermediate training states. 

There has been tremendous progress recently in  (see Sec.~\ref{sec:related} for the literature review), especially for 
 neural networks \cite{DBLP:conf/ccs/AbbaszadehPK024,zkpot:rnn,trustless-audits}, 
 logistic regression \cite{DBLP:conf/ccs/GargGJMMPW23,zkpot:adria} and 
 ordinary decision trees \cite{DBLP:conf/ccs/Pappas024},
enabling a number of important applications: 
(i) \emph{trustworthy ML-as-a-service}, where clients obtain assurance of training integrity; 
(ii) \emph{decentralized ML}, where smart contracts can reward provable training on public data; and
(iii) \emph{compliance with data restrictions}, where clients verify that models were trained on approved data sources (e.g., census data) or that certain protected content (e.g., copyrighted material) was excluded \cite{trustless-audits}.
The last application is particularly powerful: by combining zkPoT with additional ZKP for data-compliance constraints, it prevents providers from arbitrarily manipulating data.


In this work, we turn our attention to zkPoT for gradient boosted decision trees, and in particular the XGBoost library~\cite{DBLP:conf/kdd/ChenG16}. 
XGBoost has become one of the most widely used methods for structured data, routinely outperforming deep networks on medium-sized tabular datasets~\cite{DBLP:conf/nips/GrinsztajnOV22,DBLP:journals/corr/abs-2408-14817}, dominating ML competitions, and seeing widespread adoption in finance (especially in fraud detection) and healthcare.
As such, ensuring integrity and provenance of XGBoost models has wide-reaching implications in practice.


\subsection{Our contribution} 
\label{sec:intro:contribution}

\begin{figure}[t]
    \centering
    \includegraphics[width=0.7\linewidth]{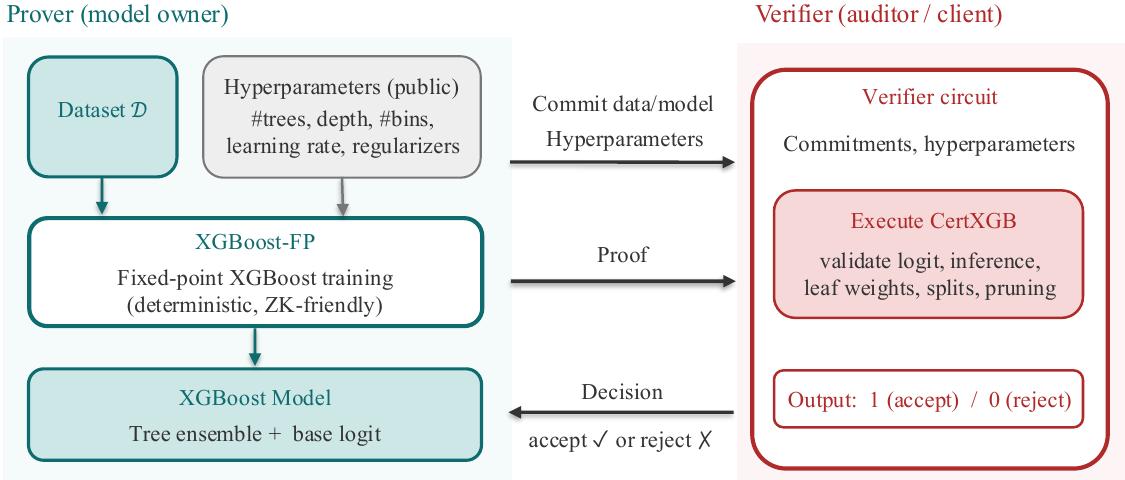}
    \caption{Overview of the ZKBoost protocol for certifying XGBoost training in zero-knowledge. To simplify, we show the owner of the model also owning the dataset but the dataset can belong to some other party in practice.}
    \label{fig:zkflow}
\end{figure}

We propose ZKBoost (illustrated in Figure~\ref{fig:zkflow}), the first zkPoT protocol allowing a model provider to prove that a classifier was correctly obtained by executing the XGBoost training algorithm on a private dataset.
We make two technical contributions, followed by a practical implementation and evaluation on standard benchmarks:

\noindent\textbf{Generic template for zkPoT of XGBoost.} 
To provide a general and optimized template for zkPoT of XGBoost, we first develop a \emph{certification algorithm} $\certxgb$ that verifies, given a resulting model $\cM$ and dataset $\cD$, that $\cM$ was correctly produced by XGBoost training---\emph{without re-executing the entire training procedure itself.}
The key insight is a fundamental restructuring of the training validation process without sacrificing accuracy. 
Naive re-execution of XGBoost training builds each tree \emph{top-down} (i.e., partitioning training samples data at each node from root to leaves) and \emph{sequentially} (with each tree depending on prior-tree predictions). 
The former requires expensive machinery if instantiated in ZKP  (e.g., oblivious data partitioning) and the latter prevents parallelization across trees.
Our $\certxgb$ instead verifies each tree \emph{bottom-up} and multiple trees \emph{in parallel}.
This restructuring decouples individual tree validation from inter-tree dependency checks, enabling parallel proof generation across trees and fundamentally changing the proof cost structure.
$\certxgb$ is general: it plugs into any ZKP backend for arithmetic circuits.
Our certification further supports validation of pre-binning and pruning (without leaking tree topology), which is crucial for practical XGBoost to prevent overfitting but unaddressed in prior zkPoT  for ordinary decision trees~\cite{DBLP:conf/ccs/Pappas024} (see Sec.~\ref{sec:related} for a comparison).

\noindent\textbf{Secure zkPoT instantiation with VOLE-ZKP.}
To instantiate zkPoT, we design a ZK protocol in which the verifier checks that the prover has faithfully executed $\certxgb$, thereby ensuring that the private tree ensemble $\cM$ was genuinely obtained by running XGBoost on the committed training data $\cD$.
Concretely, the prover first generates cryptographic commitments to $\cM$ and $\cD$, and both parties then jointly execute a state-of-the-art vector oblivious linear evaluation (VOLE)-based ZKP protocol~\cite{yang2021quicksilver,emp-toolkit}, which offers fast proving time with practical communication costs.
A subtle but critical challenge arises in this instantiation.
ZKP protocols typically emulate program execution as combinations of modular arithmetic over a finite ring, and a malicious prover may deliberately induce modular overflow in intermediate computations and exploit it  to make model $\cM$ appear valid to the verifier even if it was \emph{not} faithfully produced by XGBoost training (see Sec.~\ref{sec:zkp:nonlinear} for details).
To prevent such adversaries, we design sound ZKP subprotocols for common ML operations including \emph{comparison}, \emph{division}, and \emph{truncation}, alongside global \emph{range checks} of training traces.
These subprotocols are of independent interest, as the overflow vulnerability and our countermeasures are not specific to XGBoost and apply broadly to ZKML systems.

\noindent\textbf{Evaluation of ZKBoost and Fixed-point XGBoost.}
We implement ZKBoost in C++ and evaluate it on three real-world datasets: breast cancer, credit default, and covertype.
As no prior zkPoT for full XGBoost exists, we also compare against Sparrow~\cite{DBLP:conf/ccs/Pappas024}---the closest related work, which supports ordinary decision trees and random forests---by instantiating a simpler version of our protocol for random forest training.
Despite providing stronger security guarantees (tree-topology hiding and more rigorous soundness checks), our scheme is $3$--$6\times$ faster than Sparrow across dataset sizes.
Our implementation also resolves a gap between standard XGBoost and a program representation amenable to ZKP. 
Standard XGBoost relies on floating-point arithmetic, which is doubly problematic for zkPoT: it is incompatible with the finite-ring arithmetic underlying ZKP systems, and it is inherently non-deterministic---results can vary across hardware platforms, compilers, and rounding modes, so there is no unique training trace for the verifier to check against.
As a side contribution, we therefore implement and evaluate a \emph{fixed-point} version of XGBoost, building on standard approximation techniques from the privacy-preserving ML literature~(e.g., \cite{ESORICS:CatDeH10,EPRINT:JGBCGS23}).
We show empirically that fixed-point XGBoost matches standard floating-point XGBoost to within $1\%$ accuracy across all datasets and hyperparameter configurations.



    \section{Preliminary and Problem Statement}
\label{sec:prelims}



\smallskip
\noindent\textbf{Data Structures and Notation}
We use the bracket notation $[n]$ to denote the set $\{1,2,\ldots,n\}$ for a positive integer $n$.
Datasets are collections $\cD=\{(\vecx_i,y_i)\}_{i=1}^n$ with feature vectors $\vecx_i\in\mathbb{R}^d$ and binary labels $y_i\in\cY =\{0,1\}$.
In this paper, we use the following notation: 
$n$ as the number of data points in $\cD$; 
$m$ as the number of trees (i.e. weak learners) in $\cT = \{T_k\}_{k\in[m]}$;
$d$ as the number of features;
$B$ as the number of bins;
$h$ as the height of $T_k$;
$N=2^{h}$ as the number of leaves in $T_k$, assuming $T_k$ is a full binary tree (containing dummy nodes as explained in \ref{sec:cert}) with height $h$;\footnote{We focus on equal-width pre-binning and a canonical tree structure, which simplifies certification but differs from some production XGBoost features (e.g., quantile histograms or specialized missing-value handling).}
$N-1$ as the number of non-leaf nodes in $T_k$;
$\gamma,\lambda$ as the regularization parameters;
$\eta$ as the learning rate.
To index various quantities, we use
 $k$ to index trees;
 $i$ to index data points;
 $f$ to index features;
 $b$ to index bins;
 $\ell\in[2N-1]$ to index nodes in a tree;
 $l\in[N]$ to index leaf nodes in a tree;
 $j$ to index others (e.g., level of a node).
With this notation, each tree $T_k$ is represented as a tuple $(\vecf_k, \vectt_k, \vecw_k)\in[d]^{N-1} \times \RR^{N-1} \times \RR^{N}$. 
For each non-leaf node $\ell\in[N-1]$, $f_{k,\ell}$ indicates the feature index used for splitting, and $t_{k,\ell}$ indicates the threshold.
Each leaf node $l\in[N]$ contains a weight $w_{k,l}$. 
For each data point $\vecx_i$, $T_k$ classifies it into a leaf node $l_{k,i}$ by traversing the tree from the root to a leaf node according to the feature values of $\vecx_i$.


\subsection{XGBoost}\label{sec:xgb-overview}
We focus on the high-level training workflow.
We defer more comprehensive background to Section~\ref{app:xgb} and the original paper \cite{DBLP:conf/kdd/ChenG16}.
A decision tree maps an input feature vector $\vecx\in\mathbb{R}^d$ to a prediction by recursively partitioning the feature space using threshold-based splits and assigning a value at each leaf.
While such models are simple and interpretable, a single tree often lacks sufficient expressive power.
Gradient boosting addresses this limitation by constructing an \emph{ensemble} of trees $\cT = \{T_k\}_{k\in[m]}$ sequentially, where each new tree is trained to correct the residual errors of the current model.
XGBoost~\cite{DBLP:conf/kdd/ChenG16} is a widely used and optimized implementation of this paradigm, relying on first and second order loss derivatives to guide split selection and leaf-weight assignment across multiple boosting rounds. As a result, each stage of training depends on the predictions produced by all previous trees, making the training process inherently stateful.
This dependency highlights why certifying XGBoost training is substantially more challenging than certifying inference or the construction of a single decision tree.

\subsection{Zero-Knowledge Proof of Training}

\noindent\textbf{Zero-Knowledge Proofs and Commitments.}
ZKP (of knowledge) allows one party (a \emph{prover}) to prove knowledge of a \emph{secret witness} $\wit$ for a \emph{public statement} $\stm$ to another party (\emph{verifier}). 
Such proofs are constructed for a concrete \emph{NP relation} $\rel$, describing the relationship between $\stm$ and $\wit$. 
Formally, $\ZKP$ for an NP relation $\rel$ is a tuple of interactive Turing machines $(\cP,\cV)$, where $\cP$ is prover and $\cV$ is verifier.
Then $\cP$ and $\cV$ interact with each other, where both $\cP$ and $\cV$ take $\stm$ as common inputs, and $\cP$ additionally takes $\wit$ as a private input. 
At the end of interaction, $\cV$ outputs a binary $b$.
Proof systems that are used in verifiable ML typically require the following security properties: For an NP relation $\rel$, they must provide \emph{completeness} (i.e., if prover and verifier follow the protocol with input $(\stm,\wit)\in\rel$, verifier always accepts), \emph{knowledge soundness} (i.e., if verifier accepts, then it must be that prover owns a valid witness $\wit$ satisfying given NP relation w.r.t. statement $\stm$), and \emph{zero knowledge} (i.e., the transcript of the interaction between the prover and the (malicious) verifier leaks nothing except that there exists a witness $\wit$ such that $(\stm, \wit) \in \rel$). 
ZKP is typically paired with cryptographic \emph{commitments} to secret witness $\wit$, 
ensuring that the prover is bound to a specific value while keeping it hidden from the verifier.
Formally, a commitment scheme is defined as an algorithm $\commit$ that allows committing to a message $m$ with randomness $r$: $c\gets\mathsf{Com}(m;r)$.
As these are standard building blocks in cryptography, we refer the reader to \cite{Goldreich2001Foundations} for more details.

\noindent\textbf{Zero-Knowledge Proof of Training (zkPoT).}
Once $\ZKP$ and commitment schemes with the above properties are given, we can define a secure zkPoT for a training algorithm $\trainfp$ that takes a dataset $\mathcal{D}$ as input and outputs a model $\mathcal{M}=\trainfp(\mathcal{D})$.
For training verification, we define the relation

\newcommand{\relxgb}{R_{\text{xgb}} }
\begin{align}
\small
\relxgb = \setpred{((\cC_\cM, \cC_\cD), (\cM,\cD,r,\rho))}{\cC_\cM = \commit(\cM;r); \cC_\cD = \commit(\cD;\rho); \cM = \trainfp(\cD)}    \label{eq:rel}
\end{align}
A zkPoT protocol ensures that: (1) Completeness: An honest prover with dataset $\mathcal{D}$ and model $\mathcal{M}=\trainfp(\mathcal{D})$ can produce a valid proof. (2) (Knowledge) Soundness: Any prover that outputs a valid proof must ``know'' such a dataset $\cD$ and a valid model $\cM$ derived from $\cD$ via $\trainfp$. (3) Zero-Knowledge: The proof leaks no information about $\mathcal{D}$ or $\mathcal{M}$ beyond the commitments.
To formalize the above properties, we use the so-called \emph{simulation-based security} defined in terms of a commit-and-prove ideal functionality (see Appendix~\ref{sec:cpzkp} for details). 

    \section{Design of zkPoT for XGBoost}
\label{sec:overview}
\subsection{zkPoT Template for XGBoost Training}
\label{sec:cert}
We first introduce our certification algorithm $\certxgb$ for verifying the correctness of XGBoost training, and use it to describe a generic zkPoT template (Protocol~\ref{prot:zkp-training-short}) for relation $\relxgb$ \eqref{eq:rel}.
Essentially, $\certxgb$ represents a circuit that 
takes a dataset $\cD = \{(\vecx_i,y_i)\}_{i=1}^n$ and a claimed model $\cM = (\cT, z_0)$ as input, and then outputs 1 (``accept'') or 0 (``reject'') by checking that the model $\cM$ is correctly produced by training on $\cD$. 
The complete procedure is deferred to Appendix~\ref{sec:xgboost-circuit}.

\noindent\textbf{Overview} 
Recall that the original XGBoost (Appendix~\ref{app:xgb}) essentially proceeds as follows: 
\begin{enumerate}[leftmargin=*]
\item Initialize the base score (logit) $z_0$ from $\cD$.
\item \label{step:tree-build} \textbf{Iterative Tree Building.} For each boosting round $k=1,\ldots,m$, it grows a new tree $T_k$ \emph{from root to leaves} by (a) computing gradients and Hessians, (b) finding the best splits maximizing the $\gain$, and (c) computing leaf weights. Finally, it computes the updated scores $\vecz_k = (z_{k,i})_{i=1}^n$.
\end{enumerate}

Instead of naively re-executing the above training procedure, $\certxgb$ takes advantage of the fact that the final model $\cM = (\cT, z_0)$ is provided as input.
At a high level, $\certxgb$ verifies that each tree $T_k$ in $\cT$ is correctly built by shuffling the order of operations in Step~\ref{step:tree-build} above: it first evaluates $T_k$ on every data point $\vecx_i$ to determine the reached leaf indices $l_{k,i}$ and update the score $z_{k,i}$, and then reconstructs the histograms of gradients and Hessians at each node using these indices and scores \emph{from leaves to root}.
Finally, it checks that each split in $T_k$ maximizes the gain computed from the reconstructed histograms, and that each leaf weight is consistent with the corresponding histogram.
In slightly more detail, $\certxgb$ verifies that $(\cT, z_0)$ is the output of $\trainfp$ on input $(\vecx, \vecy)$, following the steps below. 
Note that Step~\ref{step:tree-validation}, dominating the validation costs, can be parallelized, because each tree validation is independent of others and inter-tree dependencies can be handled separately in Step~\ref{step:initialize-scores}.
See Section~\ref{app:fixedpointxgb}.
The reader may jump to the detailed procedure by clicking each subroutine.
\begin{enumerate}[leftmargin=*]

\item 
$\validatelogit$:
To validate the base score $z_0$, $\certxgb$ computes the initial prediction probability $p = \frac{1}{n}\sum_{i=1}^n y_i$ and then converts it to the logit $z'_0=\sigma^{-1}(p)=\log(p/(1-p))$.
The circuit asserts that $z_0 = z'_0$ and sets $z_{0,i} = z_0$ for all $i\in[n]$.

\item \label{step:initialize-scores} \textbf{Initialize intermediate scores.}
For each tree $k\in[m]$ and each data point $i\in[n]$, $\certxgb$ evaluates the tree $T_k=(\vecf_k,\vectt_k,\vecw_k)$ on $\vecx_i$ to compute the reached leaf $l_{k,i}\in[N]$, and updates the score $z_{k,i} \gets z_{k-1,i}-\eta\cdot w_{k,l_{k,i}}$. 

\item \label{step:tree-validation} \textbf{Tree Validation.} For each tree $k\in[m]$, do:
\begin{enumerate}[leftmargin=*]
\item $\validateinference$: 
Validate that the computed leaf indices $\vecl_k = (l_{k,i})_{i=1}^n$ are correct by checking that each $l_{k,i}$ is consistent with the feature splits $(\vecf_k,\vectt_k)$ and the input $\vecx_i$.

\item 
$\inithists$: 
Computes histograms $(G_k,H_k)$ by aggregating gradients and hessians over bins and leaves using the bin indices and leaf indices $\vecl_k$.
Here, when computing the gradients $g_i = p_i - y_i$ and hessians $h_i = p_i(1-p_i)$, $\certxgb$ uses the prediction probabilities computed via sigmoid: $p_i =\sigma(z_{k-1,i}) = 1/{(1 + e^{-z_{k-1,i}})}$.
\item 
$\validateleafweights$:
Validate the leaf weights $\vecw_k$ by checking consistency with the histograms $(G_k,H_k)$.

\item 
$\validatesplits$:
Validate the splits $(\vecf_k,\vectt_k)$ by ensuring (1) the gain derived from each split is maximal, \emph{or} (2) if the gain is $\leq 0$ (i.e., pruning condition is met), the corresponding node contains dummy values. 
\end{enumerate}
\end{enumerate}

In Appendix \ref{sec:xgboost-circuit}, we prove the following claim:

\noindent\textbf{Claim}~\ref{claim:cert-xgboost} \textit{The certificate algorithm $\certxgb$ verifies the correctness of XGBoost training.
	That is, for any $\cD=(\vecx,\vecy), \cT = \{T_k\}_{k\in[m]}$ and $z_0$, $\certxgb(\vecx,\vecy,\cT,z_0) = 1$ if and only if $(\cT,z_0)=\trainfp(\vecx,\vecy)$.
}

\noindent\textbf{Pruning without leaking tree structure}
XGBoost prunes internal nodes (marking them as leaves) to prevent overfitting when the maximum gain is $\leq 0$; revealing which nodes are pruned leaks tree topology to the verifier. 
We replace pruned nodes with dummy values $(\fdummy,\tdummy)$: $\fdummy\in[d]$ can be an arbitrary feature, and $\tdummy$ is set to a vacuous constant (e.g., $\texttt{DBL\_MIN}$) so that any sample goes right.
$\validatesplits$, if instantiated with ZKP, can then check obliviously that dummies are used exactly when the pruning condition holds.

\noindent\textbf{Fixed-point arithmetic and approximation of non-linear operations.} While the above procedure describes the high-level logic of $\certxgb$, 
the actual implementation involves careful handling of fixed-point arithmetic and approximations to ensure that all computations are sound and efficient in a proof system. 
For example, the sigmoid $\sigma(z_{k-1,i})$ is implemented as a piecewise-linear approximation, and all divisions are reformulated to avoid floating-point operations. 
The complete details are provided in Appendix~\ref{app:fixedpointxgb}.


\newcommand{\relxg}{\cR_{\text{xg}}}


\vspace{-5pt}
\begin{protocol}[label={prot:zkp-training-short}]{zkPoT for XGBoost}

\textbf{Parameters}: Number of points $n$, features $d$, trees $m$, bins $B$, leaves $N=2^h$ at depth $h$, learning rate $\eta$, regularizers $\lambda,\gamma$.

\textbf{Public Input}: Auxiliary commitments $\com_{\cD}$ to training data $\cD = (\vecx,\vecy)$, $\com_{\cT}$ to trees $\cT=\{T_k\}_{k\in[m]}$ with $T_k=(\vecf_k,\vectt_k,\vecw_k)$, and $\com_{z}$ to base logit $z_0$.

\textbf{Private Input of $\cP$}: Fixed-point training data $(\vecx, \vecy)$, the fitted tree ensemble $\cT$, auxiliary traces, and commitment randomness.

\medskip
\noindent\textbf{Validate Commitments.} $\cP$ and $\cV$ run a subprotocol for proving knowledge of commitment openings: $\com_\cT=\commit(\cT;r_\cT), \com_z=\commit(z_0;r_z), \com_\cD = \commit(\cD;r_\cD)$.

\medskip
\noindent\textbf{Validate Binary Input.} $\cP$ and $\cV$ run a subprotocol for asserting $y_i\in\{0,1\},\forall i\in[n]$ 

\medskip
\noindent\textbf{Compute Edges and binID.}
$\cP$ and $\cV$ run a subprotocol for $(\edges,\binID)\gets\prebinfp(\vecx,B)$, where $\edges$ are bin left edges and $\binID$ are bin indices for each feature of each sample; expose authenticated $(\edges,\binID)$ to be reused throughout. 

\medskip
\noindent\textbf{Validate Base Logit.}
$\cP$ and $\cV$ run a subprotocol for $\validatelogit(\vecy,z_0)=1$ to validate the initial logit $z_0=2\,\mathrm{atanh}(2p-1)$ with $p=\frac1n\sum_i y_i$. They then define the replicated vector $\vecz_0 = (z_{0,i})_{i=1}^n$ with $z_{0,i} = z_0$.

\medskip
\noindent\textbf{Initialize Intermediate Scores.}
For each tree $k\in[m]$ and data point $i\in[n]$, $\cP$ computes the reached leaf $l_{k,i}\in[N]$ by evaluating $T_k$ on $\vecx_i$, updates the score $z_{k,i} \gets z_{k-1,i}-\eta\cdot w_{k,l_{k,i}}$, and defines $\vecz_k \gets (z_{k,i})_{i=1}^n$.
Let $\vecl_k = (l_{k,i})_{i=1}^n$.

\medskip
For each tree $k\in[m]$, perform the following: 

\medskip
\noindent\textbf{Validate Inference.}
$\cP$ and $\cV$ run a subprotocol for $\validateinference(\vecx,\vecf_k,\vectt_k,\vecl_k)=1$ to validate $\vecl_k$.

\medskip
\noindent\textbf{Initialize Node Histograms.} 
$\cP$ and $\cV$ run a subprotocol for $(G_k,H_k)\gets\inithists(\vecz_{k-1},\vecy,\binID,\vecl_k)$, where $G_k[f][\ell][b]$ and $H_k[f][\ell][b]$ are the sum of gradients and hessians of samples in node $\ell$ whose feature $f$ falls into bin $b$.
These are computed by summing over samples in each leaf $l$ and propagating them up to the root.

\medskip
\noindent\textbf{Validate Leaf Weights.}
$\cP$ and $\cV$ run a subprotocol for $\validateleafweights(G_k,H_k,\vecw_k)=1$ to validate all leaf weights $\vecw_k$.

\medskip
\noindent\textbf{Validate Tree Splits.}
$\cP$ computes $(\vecf_k,\vectt_k)\gets (f^*_{k,\ell},t^*_{k,\ell})_{\ell\in[N-1]}$, where $(f^*_{k,\ell},t^*_{k,\ell})$ for $\ell$-th internal node is a split maximizing the gain derived from $G_k,H_k$.
$\cP$ and $\cV$ run a subprotocol for $\validatesplits(G_k,H_k,\vecf_k,\vectt_k,\vecf^*_k,\vectt^*_k,\edges)=1$ to check that (1) the gain derived from each split in $(\vecf_k,\vectt_k)$ is maximal, \emph{or} (2) if the gain is $\leq 0$, the corresponding node contains dummy values. 
\end{protocol}

    \subsection{ZKP Gadgets for Nonlinear Relations}
\label{sec:zkp:nonlinear}


In Sections \ref{sec:cert} and \ref{app:fixedpointxgb} we build certification relation tailored to XGBoost and its fixed-point implementation ($\trainfp$).
Next, we tackle the proof of low-level nonlinear functions that dominate the overhead of zkPoT for $\trainfp$.
Denote $\auth{\cdot}$ as the witnesses committed by $\cP$, the proof needs  
\begin{itemize}[leftmargin=*]
    \item The proof of comparison between two signed fixed-point numbers, denoted as $\auth{s} \leftarrow \bm{1}\{\auth{x} < \auth{y}\}$. It proves $s = 1$ if $x < y$, and $s = 0$ otherwise. 
    \item The proof of division between two fixed-point numbers, denoted as $\auth{z} \leftarrow \auth{x} / \auth{y}$.
    \item The proof of truncation after the multiplication of two fixed-point numbers, denoted as $\auth{y} \leftarrow {\sf Trunc}(\auth{x})$, where $\auth{y} = \lfloor \auth{x} / 2^f \rfloor$ assume $f$-bit precision.
    \item The proof of histogram construction denoted as $[T] \leftarrow {\sf Histogram}(\auth{\bm{a}}, \auth{\bm{v}})$ which sets the $i$-th histogram bucket $\auth{T_i} = \sum v_j \cdot \bm{1}\{a_j = i\}$.
\end{itemize}
Existing approaches do not satisfy our goal due to efficiency, vulnerability or lack of support for fixed-point arithmetic. We design rigorous constraint systems to verify these operations with the focus of both efficiency and soundness. We sketch our main contribution in a high level and defer the formalized description and the instantiation of other components to Appendix~\ref{sec:app:zkxgboostfp}.



\noindent\textbf{Comparison.}
We improve the idea of~\cite{hao2024scalable} to prove the comparison relation by bits-decomposition. Note that we use a signed representation. Hence, we have $\bm{1}\{\auth{x} < \auth{y}\} = {\sf MSB}(\auth{x-y})$, where {\sf MSB} refers to the most significant bit of $z = x - y$. The main task is shifted to proving $\auth{s} = {\sf MSB}(\auth{z})$ given committed $(\auth{z},\auth{s})$. 
We describe prior approach in detail in Appendix~\ref{sec:app:comparison}, and focus on improving its consistency checks on invalid bits-decomposition.


We abstract out the problem in the following way. The proof of bit-decomposition requires $\cP$ to commit to groups of bits $({z}_0,\dots,{z}_{t-1})$ each of $d$-bit, and show that $z = 2^{n-1} \cdot s + \sum_{i=0}^{t-1} 2^{id} \cdot {z}_i$. However, a cheating $\cP$ may commit to incorrect $(\tilde{z}_0,\dots,\tilde{z}_{t-1})$ such that
$z + p = 2^{n-1} \cdot s' + \sum_{i=0}^{t-1} 2^{id} \cdot \tilde{z}_i$, which would completely flip the MSB $s$. 
In the context of zkPoT for XGBoost, this is particularly damaging: comparison results determine whether one split gain is greater than another, so a flipped MSB could allow a cheating $\cP$ to commit to a tree split that does \emph{not} achieve the maximum gain---passing verification while certifying a model that deviates from the prescribed training algorithm.

Our solution employs a Mersenne prime in the form of $p = 2^n - 1$. The only chance a cheating $\cP$ has is when $z=0$, it may claim $s=1$ and $\tilde{z}_i = 2^d-1$ for all $i \in [0,t)$.
By defining $\auth{w} = \bm{1}\{\auth{z} \neq 0\}$, we observe that
$
{\sf MSB}(z) = \begin{cases}
   0 &\text{if } z = 0, w = 0 \\
   s &\text{if } z \neq 0, w = 1
\end{cases}
$.
It implies that $(1 - w) \cdot s = 0$ should always be true. Additionally, we employ the non-equality-zero check from~\cite{SP:PHGR13} to prove $\auth{w} = \bm{1}\{\auth{z} \neq 0\}$. Our solution to address the soundness issue only requires committing 2 extra values and proving 3 multiplicative relations, while the previous work almost costs $2\times$~\cite{hao2024scalable}. 

\noindent\textbf{Division and Truncation.} Proof of division takes inputs $(\auth{x},\auth{y},\auth{z})$ and proves that $z = \lfloor x / y \rfloor$. 
In our $\trainfp$, the numerators and denominators may range from $0$ to nearly $p/2^f$. Previous proofs of division are not suitable since they either only work for bounded small values~\cite{CCS:LiuXieZha21,DBLP:conf/ccs/Pappas024} or requires thousands of logical constraints~\cite{weng2021mystique}.

To prove the relation, we leverage the existence of a residue $r \in [0,y)$ such that $x = z\cdot y + r$. Additionally, two range checks are required to maintain the soundness: $r \in [0,y)$, and $z \in [0,\lfloor p/y \rfloor]$. The first check ensures $r$ is a residue. The second check is needed to prevent the similar wrap-around issue happened in the proof of MSB. Namely, there could be many possible pairs of $(z', r')$ such that $x = z' \cdot y + r' \; {\sf mod} \; p$. The proof should ensure a $z$ that satisfies $x = z \cdot y + r$ without modulo $p$.
The truncation is a simplified division with public and positive denominators $y = 2^f$. 



\noindent\textbf{Proof of Histogram Construction.} Our proof commits to a vector $\{(a_j,v_j)\}_{j\in[N]}$ and a histogram $T = \{T_i\}_{i\in[n]}$, and proves $\auth{T_i} = \sum v_j \cdot \bm{1}\{a_j = i\}$. A naive approach following the previous proof of random forest~\cite{DBLP:conf/ccs/Pappas024} is to utilize a ZK-RAM: for $j \in [N]$, prove a read at location $a_j$, increment its value by $v_j$, and write it back. It incurs $O(N+n)$ overhead with a large constant~\cite{yang2024two}.

We observe that the proof of histogram construction is simpler than ZK-RAM since it only consists of \emph{write} operation and the operator is public. Hence, it can be proven by a special weighted LogUp proof by showing an equation $\sum_{j\in[N]} \frac{v_j}{X + a_j} = \sum_{i\in[n]} \frac{T_i}{X + i}$ is correct with respect to a variable $X$~\cite{DBLP:journals/iacr/Habock22a}. In our implementation, this reduces the overhead by around $3\times$.

\noindent\textbf{Global Range Checks for Soundness.} To ensure the soundness, global range check is needed for every intermediate values that are committed during the proof. This is because a cheating prover may leverage the overflow to cause inconsistency in numerical operations. We carefully analyze the range of witnesses and implement the range checks throughout the protocol. 

\noindent\textbf{Instantiation and Security.} Our ZKP protocol can be instantiated by any general-purpose proof system. To achieve the maximum prover efficiency, we realize it in the VOLE-ZK framework~\cite{yang2021quicksilver}. We state our main theorem of security below and provide details and security analysis in Appendix~\ref{sec:app:zkp:security}.

    

\noindent\textbf{Theorem 1.} \textit{Define the relation by the proof of fixed-point XGBoost Training described in $\relxgb$~\eqref{eq:rel},  Protocol \ref{prot:zkp-training-short} securely realizes $\Fcp$ in the $(\Func{CVOLE},\Func{ZK})$-hybrid model.}

    \section{Experiments}
\label{sec:experiments}

In this section, we first analyze how fixed-point representation affects the accuracy and efficiency of the XGBoost training in plaintext. 
Secondly, we benchmark the performance of our ZKBoost implementation in terms of the proving time, bandwidth usage,  and memory usage, and compares it with the baseline scheme. 

\subsection{Fixed-point XGBoost}

\noindent\textbf{Setup.}
We compare our fixed-point arithmetic gradient boosting implementation (``Fixed'') against \textsc{XGBoost} (scikit-learn API), using identical hyperparameters (depth, \texttt{bin}=128, $\eta= 0.3$, $\lambda=1$, $\gamma=0$) and identical train/test splits per configuration. Each result reported is an average over 10 runs and we observed negligible variance. We evaluate \emph{equal-thread} runs with both OpenMP and BLAS pinned to one thread. 

\noindent\textbf{Datasets.}
We use four standard binary classification benchmarks: \textbf{Breast Cancer} $(n=569, d=30)$, \textbf{Default of Credit Card Clients} $(n= 30001, d=23)$, two \textbf{Covertype} subsets (converted for binary classification), \textbf{Covertype 50k} $(n=50000, d=54)$ and \textbf{Covertype 100k} $(n=100{,}000, d=54)$, and \textbf{Adult} ($n=45222$, $d=104$). We abbreviate the first three to \textbf{BR}, \textbf{CR}, \textbf{CO}.

\noindent\textbf{Accuracy parity.}
Across all benchmarks and the tested hyperparameter grid, our fixed-point training \emph{tracks} floating-point XGBoost to within statistical error. Every configuration satisfies $\accgap \le 1\%$. Notably, even on deeper/larger ensembles that are not presented here, the maximum observed gap remains under $0.01$. These results confirm that adopting fixed-point arithmetic, together with our logit initializer, piecewise sigmoid, and fixed-point gain computation, does not degrade predictive accuracy in practice.
Due to space constraints we defer the tables to Appendix~\ref{app:exptab} and the detailed performance discussion in Appendix~\ref{app:expxgboost}.


\subsection{Zero Knowledge XGBoost-FP}
\label{zkexp}

\noindent\textbf{Implementation and Setup.} We implement the ZK XGBoost-FP 
in C++ over the VOLE-ZK framework~\cite{yang2021quicksilver} based on EMP-toolkits~\cite{emp-toolkit}. All experiments only employ 1 thread and report end-to-end running time. The prover and verifier are each hosted by a AWS EC2 m5.2xlarge instances located in the same region. Each is equipped with 8 vCPUs and 32GB RAM. To simulate various network condition, we use the Linux tc tool to configure the bandwidth and latency. Specifically, we emulate a local area network (LAN, 5Gbps bandwidth) and a wide-area network (WAN, 1Gbps bandwidth and 60ms round-trip latency).

\smallskip
\noindent\textbf{Building Blocks.} We improved the proof of comparison (against wrap-around attack) and histogram construction. The former makes $5.3\times$ improvement in LAN and $2.2\times$ in WAN. The latter makes $4\times$ improvement in LAN and $3\times$ in WAN. Details in Appendix~\ref{sec:appendix:perf:zkp}.

\smallskip
\noindent\textbf{Benchmarking Datasets.} We benchmark the overhead of proving the XGBoost training over datasets ${\sf BR, CR, CO}$ and show the results in Table~\ref{tab:overhead:datasets} (left). With the large 60ms latency, the WAN running time is $83\%$ to $87\%$ more than it LAN due to VOLE interactions. However, all proofs are generated in a reasonable amount of time. Note that as demonstrated in our accuracy tests, training with $50$ to $100$ trees already achieves high accuracy. Hence, the results in Table~\ref{tab:overhead:datasets} can be viewed as an upper bound of runtime in the real world. The overall performance on general-purpose servers also demonstrates that our lightweight scheme is suitable for commodity hardware.

\begin{table}[t]
  \centering
  \footnotesize
    \caption{\small \textbf{Left}: Running time in LAN and WAN (minutes), communication overhead (GB), and memory usage (GB). Three datasets are ${\sf BR, CR, CO}$. With 100 trees of depth 5. \;\;\;\;
    \textbf{Right}: Running time for the proof of random forest training (minutes). Parameters are dataset size $n\in\{2^{14},2^{16},2^{18}\}$, $h=5,m=64,d=16$. Sparrow's bagging time ($30\%$) is taken out.}
  \begin{tabular}{@{}ccccc || cccc@{}}
  \toprule
  & LAN & WAN & Communication & Memory & $n$ & $2^{14}$ & $2^{16}$ & $2^{18}$ \\\midrule
  ${\sf BR}$ & 18.65 & 31.93 & 52 & 0.284 & Sparrow & 17.64 & 45.01 & 155.19 \\
  ${\sf CR}$ & 21.81 & 37.64 & 61 & 0.433 & Ours, LAN & 5.52 & 10.01 & 25.96 \\
  ${\sf CO}$ & 52.53 & 93.15 & 154 & 0.710 & Ours, WAN & 9.76 & 18.38 & 52.69 \\
    \bottomrule
  \end{tabular}
  \label{tab:overhead:datasets}
\end{table}

\smallskip
\noindent\textbf{Running Time with Respect to Number of Trees.} We demonstrate the relation between the running time and the number of trees ranging from $25$ to $200$. We use the largest ${\sf CO}$ dataset with 40k data points, 54 features, and tree-depth 5. The result in Figure~\ref{fig:perf} (left) shows that the running time is strictly linear to the number of trees. We use 1 thread for consistency with common ZKP benchmarks. As discussed in Sec.~\ref{sec:intro:contribution}, our scheme is friendly to multi-threading and can parallelize the proving of each tree.

\begin{figure}[t]
    \centering
    \includegraphics[width=0.7\linewidth]{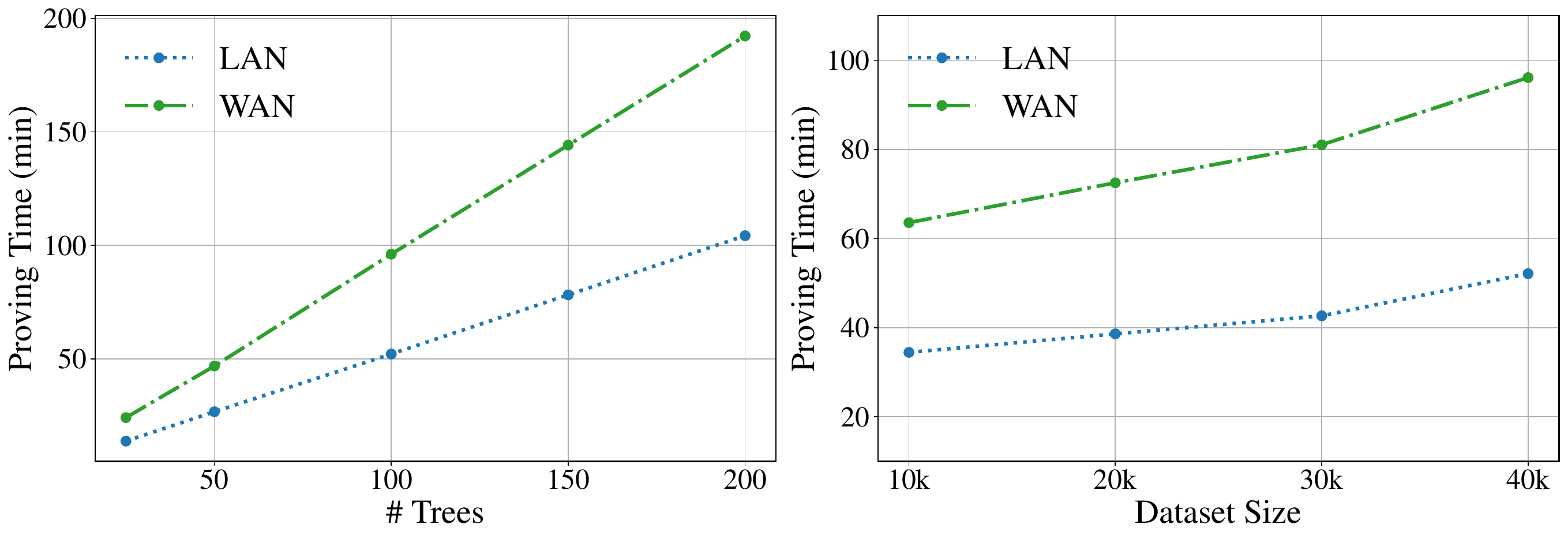}
    \caption{Benchmarking ${\sf CO}$ with increased number of trees (left) or data points (right).}
    \label{fig:perf}
\end{figure}

\smallskip
\noindent\textbf{Running Time with Respect to Dataset Size.} We benchmark training using a high-dimension dataset ${\sf CO}$ with increasing number of data points, ranging from 10k to 40k. The result in Figure~\ref{fig:perf} (right) shows that the running time is not significantly impacted by the number of data points. This is due to the efficient proof of histogram construction mentioned in Section~\ref{sec:zkp:nonlinear}, which eliminates the bottleneck on arranging data points to histogram cells. It also demonstrate that our scheme can prove the training on much larger datasets without significant degradation on performance.

\smallskip
\noindent\textbf{zkPoT for Random Forest.} Our scheme is the only ZKP for XGBoost training. To compare with a baseline, we simplify it to prove random forest training and compare with Sparrow~\cite{DBLP:conf/ccs/Pappas024}. The detailed approach is deferred to Appendix~\ref{sec:forest-circuit} and we provide a qualitative comparison in \ref{sec:related}. 
Sparrow is a memory-efficient zk-SNARK which supports public verifiability, our protocol is interactive. However, our scheme provides a few new features. Firstly, it hides tree topology by filling a full binary tree with dummy nodes for better privacy. Secondly, it enforces more rigorous consistency checks for nonlinear proofs (see Section~\ref{sec:zkp:nonlinear}). 
These additional guarantees lead to a larger circuit and, consequently, higher proof complexity. 

Additionally, we observe that the ZK-RAM used by Sparrow for proving the read/write operation to histograms can be simplified to a proof of subset. It can also be viewed as a special case of our proof of histogram construction with the \emph{write} value to always be 1. This method helps our scheme gain significant advantage since building the leaf-layer histogram is the main bottleneck for large datasets.

Table~\ref{tab:overhead:datasets} (right) shows the running time of our proof of random forest training compared to Sparrow~\cite{DBLP:conf/ccs/Pappas024}. Our scheme achieves better performance across all tests in both LAN and WAN with up to $6\times$ faster in running time.

    \section{Conclusion}
\label{sec:conclusion}

We introduced the first zkPoT for gradient boosted decision trees, realized concretely for XGBoost. 
Our current instantiation targets binary classification, covering high-stakes applications such as fraud detection, credit risk prediction, and mortality prediction.
Two natural directions for future work are (i) extending to multi-class XGBoost, which is conceptually straightforward via our framework (see Appendix~\ref{sec:multi-class}), and (ii) replacing the interactive VOLE-based instantiation with a non-interactive succinct proof system to enable public verifiability at the cost of higher proving time.

	\ifAnon
    \else
    \section*{Acknowledgments}
    We would like to thank Daniel Escudero for useful discussions during the early stages of this project. Nikolas Melissaris is supported by ERC grant OBELiSC (101115790).

    \section*{Disclaimer}
    This paper was prepared in part for information purposes by the Artificial Intelligence Research group of JPMorgan Chase \& Co and its affiliates (“JP Morgan”), and is not a product of the Research Department of JP Morgan. JP Morgan makes no representation and warranty whatsoever and disclaims all liability, for the completeness, accuracy or reliability of the information contained herein. This
    document is not intended as investment research or investment advice, or a recommendation, offer or solicitation for the purchase or sale of any security, financial instrument, financial product or service, or to be used in any way for evaluating the merits of participating in any transaction, and shall not constitute a solicitation under any jurisdiction or to any person, if such solicitation under such jurisdiction or to such person would be unlawful.
    \fi
	\bibliography{manualref,cryptobib/abbrev3,cryptobib/crypto}
	\ifFull
    	\bibliographystyle{myalpha}
		\appendix
	\fi
	
	\section{Additional Preliminaries}
\label{sec:more-prelim}
\subsection{Further Details of XGBoost}
\label{app:xgb}
\subsubsection{Theoretical Background} XGBoost is a scalable implementation of gradient boosted decision trees. Given a dataset $\mathcal{D}$ and a differentiable convex loss function $\ell:\mathbb{R}\times\mathcal{Y}\to\mathbb{R}$, the algorithm trains an ensemble of $m$ regression trees $\cT = \{T_k\}_{k=1}^m$ such that the final prediction is
\[
p = \sigma(z) \qquad z = z_0 - \sum_{k=1}^m \eta \cdot w_{k,i}
\]
where $\sigma$ is the sigmoid function, $z$ is the final logit, $w_{k,i} = T_k(\vecx_i)$ is the output of tree $T_k$ on input $\vecx_i$, $\eta\in(0,1]$ is the learning rate, and $z_0$ is the initial prediction (logit). 
The training algorithm builds trees iteratively to minimize the regularized objective. 
At each boosting round $k$, the algorithm computes the first- and second-order derivatives of the loss:
\[
g_{k,i} = \frac{\partial}{\partial p_{k-1}} \ell(p_{k-1},y_i), 
\quad
h_{k,i} = \frac{\partial^2}{\partial p_{k-1}^2} \ell(p_{k-1},y_i),
\]
where $p_{k-1}$ is the prediction after $k-1$ rounds. These $(g_{k,i},h_{k,i})$ values are called \emph{gradients} and \emph{Hessians}.

\noindent\textbf{Split finding.} For each candidate split (feature $f$ and threshold $t$), the gain is computed as
\[
\textsf{Gain}(f,t) \;=\; 
\frac{1}{2}\cdot\left(\frac{G_{\cL}^2}{H_{\cL}+\lambda} + \frac{G_{\cR}^2}{H_{\cR}+\lambda} - \frac{(G_{\cL}+G_{\cR})^2}{H_{\cL}+H_{\cR}+\lambda}\right) - \gamma,
\]
where $G_{\cL}=\sum_{i\in \cL} g_i$, $H_{\cL}=\sum_{i\in \cL} h_i$, and similarly for the right child $\cR$. Hyperparameters $\lambda,\gamma$ control regularization and split pruning. The best split maximizes the gain.

\noindent\textbf{Leaf weights.} Once a tree is grown, each leaf $l$ receiving samples $\cI_l$ gets a weight
\[
w_l = \frac{\sum_{i\in \cI_l} g_i}{\sum_{i\in \cI_l} h_i + \lambda}.
\]

\noindent\textbf{Objective.} The ensemble minimizes the regularized objective
\begin{align*}
L_k &= \sum_{i=1}^n [g_{k,i} T_k(\vecx_i) + \tfrac{1}{2} h_{k,i} T_k(\vecx_i)^2] + \Omega(T_k), \\
\Omega(T) &= \gamma N + \frac{1}{2}\lambda \sum_{l} w_l^2,
\end{align*}
where $N$ is the number of leaves. Training iteratively adds trees until $m$ rounds are completed.

\noindent\textbf{Loss function.} In this work, we use the cross-entropy loss $\ell(p,y)= -y \log(p) - (1-y)\log(1-p)$, where $p=\sigma(z)\in(0,1)$ is the predicted probability. 
Then the gradients and hessians wrt logit are
\[
g = \frac{\partial \ell}{\partial z} =  p - y, \quad h = \frac{\partial^2 \ell}{\partial z^2} = p(1-p).
\]

\subsubsection{XGBoost Inference}
Computation of the prediction for a given data point $\vecx$ proceeds in three steps:
\begin{enumerate}
  \item \textbf{Tree Evaluation:} For each tree $T_k\in\cT$, compute the weight $w_{k} = T_k(\vecx)$.

  \item \textbf{Score Aggregation:} Compute the final score (or \emph{logit}) $z$:
  \[
    z = z_0 - \sum_{k=1}^m \eta \cdot w_{k}
  \]
  where $z_0$ is the base score determined during training and $\eta$ is the learning rate.

  \item \textbf{Probability Transformation:} Transform the score $z$ in the log-odds space via the sigmoid function $\sigma$ to yield the final probability:
  \begin{align}
    p = \sigma(z) = \frac{1}{1 + e^{-z}} \label{eq:sigmoid}
  \end{align}
\end{enumerate}

\subsubsection{XGBoost Training}
To train a tree ensemble $\cT$ on a dataset $\cD= \{(\vecx_i, y_i)\}_{i=1}^n$, XGBoost first initializes the base score $z_0$, pre-bins features into $B$ discrete bins, and then iteratively builds trees $T_k$ for boosting rounds $k=1,\ldots,m$ by selecting splits and computing leaf weights to maximize the gain. 
Note that $z_0$ is also part of the trained model, so the final model $\cM$ consists of $\cT$ and $z_0$.
We elaborate on these steps below: 

\begin{enumerate}
\item \textbf{Base Score Initialization.}
    It initializes the base score $z_0$ by taking the log-odds 
    \begin{align}
      z_0 = \sigma^{-1}(p) = \log\frac{p}{1 - p} \label{eq:log-odds}
    \end{align} 
    where $p = \frac{1}{n}\sum_i y_i$ is the mean label probability.

\item \textbf{Pre-binning.}
    Each feature column is discretized into $B$ equal-width bins, and a lookup table $\edges\in\RR^{d\times B}$ of left edges of bins is constructed.
    This \emph{pre-binning} step replaces direct floating-point comparisons on raw feature values with integer bin indices,
    enabling efficient and deterministic split evaluation.

\item For each boosting round $k=1,\ldots,m$, it builds a new tree $T_k$ with the following key steps.
\begin{enumerate}
    \item \textbf{Prediction and derivatives.}
    For every data point $i\in[n]$, the algorithm computes the current probability
    $p_i = \sigma(z_{k-1,i})$, the gradient $g_i = p_i - y_i$, and the Hessian $h_i = p_i(1 - p_i)$.

    \item \textbf{Split search.}
    For each node, containing samples indexed by $\cI\subseteq[n]$, the algorithm scans candidate pairs $(f, b)$ of feature and bin threshold defining partitions $\cL$ and $\cR$ such that $\cL \cup \cR = \cI$. 
    Then the best split $(f^*, b^*)$ is selected by maximizing the gain:
    \[
    \gain \;=\; 
    \frac{1}{2}\cdot\left(\frac{G_{\cL}^2}{H_{\cL}+\lambda} + \frac{G_{\cR}^2}{H_{\cR}+\lambda} - \frac{G_{\cI}^2}{H_{\cI}+\lambda}\right) - \gamma,
    \]
    where $G_{\cI}=\sum_{i\in \cI} g_i$, $H_{\cI}=\sum_{i\in \cI} h_i$, and similarly for the children $\cL$ and $\cR$. 
    Hyperparameters $\lambda,\gamma$ control regularization and split pruning.

    \item \textbf{Leaf weight computation.}
    Once the tree structure is finalized,
    each leaf $l\in[N]$ is assigned the weight
    \[
        w_l = \frac{ G_{\cI_l}}{H_{\cI_l} + \lambda},
    \]

    \item \textbf{Logit update.}
    Finally, for each data point $i\in[n]$, the logit is updated as
    \[
        z_{k,i} = z_{k-1,i} - \eta \cdot T_k(\vecx_i).
    \]

    \end{enumerate}
\end{enumerate}


\subsection{Details of Commit-and-Prove Zero Knowledge Proofs}\label{sec:cpzkp}

\begin{functionality}[label={func:cp}]{$\Fcp[\rel]$}
The functionality $\Fcp$ interacts with three parties: a prover $\cP$, a verifier $\cV$, and an ideal adversary $\cS$. It is parameterized by an NP relation $\rel$.    

\textbf{Committing Phase}
Upon receiving $(\cmd{Commit},\wit)$ from $\cP$:
  \begin{algorithmic}[1]
  	\State Store $\wit$ internally and send $\cmd{Commit}$ to $\cS$.
    \State Upon receiving $(\cmd{Do-Commit})$ from $\cS$, send $\cmd{Commit-Receipt}$ to $\cV$.
  \end{algorithmic}
\textbf{Proving Phase}
Upon receiving $(\cmd{Prove},\stm,\wit)$ from $\cP$:
  \begin{algorithmic}[1]
  	\State If $(\stm,\wit)\notin\rel$, ignore the input. 
    \State Else, send $(\cmd{Prove},\stm)$ to $\cS$.
    \State Upon receiving $\cmd{Do-Prove}$ from $\cS$, send $(\cmd{Prove-Receipt}, \stm)$ to $\cV$.
  \end{algorithmic}
\end{functionality}

\subsection{Zero-Knowledge Proof based on Vector Oblivious Linear Evaluation}\label{sec:vole-zk}

A series of works construct interactive ZKPs based on the vector oblivious linear evaluation (VOLE)~\cite{dittmer_et_al:LIPIcs.ITC.2021.5,SP:WYKW21,C:BMRS21,CCS:YSWW21,USENIX:WYXKW21,CCS:BBMRS21,C:BBMS22,hao2024scalable}. They require a preprocessing (witness-independent) phase that generates a uniformly sampled VOLE correlation $\vecw = \vecv + \vecu \cdot \Delta$, from which $\cP$ obtains $(\vecu,\vecw)$ and $\cV$ obtains $(\vecv,\Delta)$. $\Delta$ is denoted as the global key.
In the proving (witness-dependent) phase, $\cP$ and $\cV$ transform the VOLE correlation into a commitment to the witness $\vecx$, i.e. $M[\vecx] = K[\vecx] + \vecx \cdot \Delta$. $M[\vecx]$ is denoted as the message authentication code and $K[\vecx]$ is denoted as the local key. We denote the commitment to $\vecx$ as $[\vecx]$.

Proving the linear relations is \emph{free} over VOLE commitments since they are linearly homomorphic. I.e., the commitment to $[z] = [x] + [y]$ can be generated by defining $M[z] = M[x] + M[y]$ and $K[z] = K[x] + K[y]$. To prove the multiplicative relation $[z] = [x] \cdot [y]$, previous works~\cite{dittmer_et_al:LIPIcs.ITC.2021.5,CCS:YSWW21} observe that
\begin{align*}
    & K[x] \cdot K[y] + K[z] \cdot \Delta \\
    &=M[x] \cdot M[y] - (x \cdot M[y] + y \cdot M[x] - M[z]) \cdot \Delta \\
    &\; + (x \cdot y - z) \cdot \Delta^2
\end{align*}

Note that $(x\cdot y -z) \cdot \Delta^2$ is canceled if $\cP$ honestly commits to these values. Define $B = K[x] \cdot K[y] + K[z] \cdot \Delta$, $A_0 = M[x] \cdot M[y]$ and $A_1 = x \cdot M[y] + y \cdot M[x] - M[z]$, $\cP$ simply proves that it holds a valid pair of $(A_0,A_1)$ that satisfies $B = A_0 - A_1 \cdot \Delta$. This can be done by a mask-and-open. Furthermore, a large number of multiplicative relations can be proven in a batch by random linear combination.

For an arithmetic circuit of size $C$, the total communication overhead for the preprocessing phase is $O({\sf polylog} C)$ with 2 rounds and the online proving phase is $O(C)$ with 1 round. Notably, the end-to-end round complexity is $2$ and its online phase can be made non-interactive.
	\section{Related Work}\label{sec:related}
The field of zero-knowledge machine learning has evolved through two main phases: inference verification and training verification.


\emph{Comparison with Sparrow}
Sparrow \cite{DBLP:conf/ccs/Pappas024}, a zkPoT framework for ordinary decision trees using Gini impurity, lays important groundwork for our work. 
In particular, their histogram-based certification inspired our design of $\certxgb$, which initializes node histograms from leaf to root.
However, gradient boosting requires more complex inter-tree dependencies and split search than simple decision trees, which hinder naive application of Sparrow.
In addition, we address several technicalities that did not arise in Sparrow, including ZKP subprotocols for (1) pruning and binning, (2) range checks of training traces, making sure the prover did not cause arithmetic overflow to break soundness, (3) various fixed-point operations, such as division, gradients, hessians, sigmoid, etc., (4) handling corner cases such as division by 0. 
Moreover, we use VOLE-ZK as a backend for faster proving time, while Sparrow uses succinct non-interactive arguments of knowledge (SNARKs) to achieve sublinear verification time.

\emph{Foundational verifiable ML.} Early work established the feasibility of verifiable neural network execution. \cite{DBLP:conf/nips/GhodsiGG17} with SafetyNets introduced verifiable execution of deep neural networks using interactive proofs based on the GKR protocol (\cite{DBLP:journals/jacm/GoldwasserKR15}, representing DNNs as arithmetic circuits for verification with practical performance for networks with millions of parameters, though without zero-knowledge privacy.

\emph{Zero-knowledge inference verification.} Recent systems have made ZK inference practical for large neural networks. \cite{DBLP:journals/tdsc/LeeKKO24} with vCNN presented the first practical zero-knowledge proof system for CNN inference using zk-SNARKs (\cite{DBLP:conf/crypto/Ben-SassonCGTV13}), developing circuit-friendly representations of convolution and pooling operations with 8500$times$ speedup over naive approaches. \cite{cryptoeprint:2021/087} with ZEN proposed zero-knowledge neural network inference using Quadratic Arithmetic Program-based zk-SNARKs with optimizations for ReLU activations. Building on these foundations, \cite{DBLP:conf/eurosys/ChenWSK24} in ZKML developed a compiler framework that automatically translates TensorFlow models into Halo2 zk-SNARK circuits (\cite{cryptoeprint:2019/1021}), achieving orders-of-magnitude reductions in proof generation time for large models like ResNet-18 and GPT-2 through specialized ``gadgets'' for activation functions and circuit layout optimization.

\emph{Decision tree inference verification.} \cite{DBLP:conf/ccs/ZhangFZS20} formalized the first zero-knowledge proof protocols specifically for decision tree inference and accuracy testing, enabling model owners to prove tree predictions or accuracy without revealing structure or data. Their implementation demonstrates practical performance, proving accuracy of a 1,029-node tree on 5,000 test samples in about 250 seconds with 287 KB proof size. \cite{zhu2024zkdt} improved upon this work by replacing hash-based commitments with polynomial commitments, achieving over 50$\times$ reduction in multiplication gates through a ``modular'' approach that decouples tree structure from parameters.

\emph{Proof of training}
Training verification is fundamentally more challenging than inference due to the iterative, data-intensive nature of learning. \cite{DBLP:conf/ccs/GargGJMMPW23} coined the term ``zero-knowledge proofs of training'' and provided the first rigorous security definitions, demonstrating feasibility with a zkPoT protocol for logistic regression using MPC-in-the-head (\cite{DBLP:journals/siamcomp/IshaiKOS09}) techniques combined with zk-SNARKs. For deep neural networks, \cite{DBLP:journals/tifs/SunBLZ25} presented zkDL, introducing zkReLU (SNARK-friendly ReLU gradient proofs) and FAC4DNN (aggregated circuit design), achieving proof generation for an 8-layer DNN (10M parameters, batch size 64) in under 1 second per batch. \cite{DBLP:conf/ccs/AbbaszadehPK024} built Kaizen, a complete zkPoT system using Halo2 that proves the entire training process including weight updates while maintaining privacy of both data and model.

\cite{zkpot:adria} recently demonstrated a \emph{rejection sampling attack} on zkPoT protocols, allowing a malicious prover to choose the training randomness to bias the model without detection. 
As our underlying XGBoost training is deterministic, this attack does not apply to our setting.

\emph{Decision tree fairness verification.} 
Shamsabadi et al., \cite{DBLP:conf/iclr/ShamsabadiWFDGP23} introduced Confidential-PROFITT, which proves decision trees were trained under certain fairness constraints. Their approach differs from traditional training verification; instead of proving each training step, they designed a certified training algorithm where information gain on sensitive attributes is bounded, outputting both an optimal tree and a zero-knowledge proof certifying fair training without revealing data or model parameters.

\emph{Why fixed-point arithmetic?}
Garg et al., \cite{garg2022succ} showed that zero-knowledge proofs can certify IEEE-754 floating-point execution by explicitly modeling mantissa alignment, exponent comparison, normalization, and rounding at the bit level. However, each floating-point addition or multiplication expands into hundreds to thousands of Boolean or arithmetic constraints, due to bit decomposition, conditional logic, and normalization circuitry. As a result, the cost of proving a single floating-point operation is orders of magnitude larger than that of a fixed-point addition or multiplication. This overhead compounds in iterative learning workloads, where training involves millions of arithmetic operations across many boosting rounds. In contrast, our fixed-point formulation uses native field arithmetic with explicit truncation, yielding constant-overhead arithmetic operations and enabling scalable certification of gradient-boosted training.

	\section{Details of Fixed-Point XGBoost Training}\label{app:fixedpointxgb}

In this section, we provide further details of our fixed-point XGBoost training algorithm, including the complete training procedure (Procedure~\ref{proc:Train}) and numerical safeguards (Section~\ref{app:sub:safeguards}).
The notations and data structures used here are consistent with those in Section~\ref{sec:xgb-overview}.

Zero-knowledge proof systems operate over finite rings or fields, making
floating-point arithmetic impractical as it involves rounding, exponents, and platform-dependent nondeterminism that are difficult to represent succinctly. We therefore redesign XGBoost so that all quantities are represented in \emph{fixed-point arithmetic} with an implicit global scaling factor \textsf{scale}. This section introduces the representation, the proof-friendly approximations to the sigmoid and log-odds functions, and the implications for training correctness.

\smallskip
\noindent\textbf{Fixed-Point Representation.}
Each real value $x \in \mathbb{R}$ is represented by an integer
\[
  \tilde{x} = \lfloor x \cdot \textsf{scale} \rfloor,
\]
and all arithmetic is carried out over integers. Addition and subtraction operate directly on these integers, while multiplication and division are defined as:
\begin{align}
  \widetilde{x \cdot y} &= 
  \Big\lfloor \frac{\tilde{x} \cdot \tilde{y}}{\textsf{scale}} \Big\rfloor, 
  &
  \widetilde{\frac{x}{y}} &=
  \Big\lfloor \frac{\tilde{x} \cdot \textsf{scale}}{\tilde{y}} \Big\rfloor.
  \label{eq:fixed-muldiv}
\end{align}
The division by \textsf{scale} (or multiplication by it) is performed
\emph{over integers}, and the result is floored to remain within the integer domain. This keeps every operation deterministic, ensuring that the training relation forms a well-defined NP relation suitable for zero-knowledge proofs.
In the following, we omit the tilde notation and treat all variables as fixed-point integers with implicit scaling.

\smallskip
\noindent\textbf{Proof-Friendly Sigmoid and Log-Odds.}
In XGBoost, the score $z$ in the log-odds space for each boosting round gets converted into probability via sigmoid $\sigma(z)=1/(1+e^{-z})$, and the base score is conversely initialized by taking the log-odds $\sigma^{-1}(p)=\log(p/(1-p))$.
These are not proof-friendly due to the exponential and logarithmic operations, requiring floating-point arithmetic. 
We replace them with fixed-point variants that approximate these functions accurately in the range $[-2,2]$:

\begin{equation}
\sigmoidwidefp(z) =
\begin{cases}
0, & z \le -2,\\[4pt]
\dfrac{z + 2}{4}, & -2 < z < 2,\\[6pt]
1, & z \ge 2,
\end{cases}
\label{eq:sigmoidfp}
\end{equation}

\begin{equation}
\logitfromprobfp(p)
  = 2\left(u + \frac{u^{3}}{3} + \frac{u^{5}}{5}\right),
\quad u := 2p - 1.
\label{eq:logitfp}
\end{equation}

where all powers and divisions are evaluated according to
Eq.~\eqref{eq:fixed-muldiv}.  Equation~\eqref{eq:sigmoidfp} is linear within $[-2,2]$ and saturates outside this range, avoiding costly exponentials. Equation~\eqref{eq:logitfp} is a truncated Taylor expansion of the $\tanh^{-1}$ function ($\mathrm{atanh}(u) = u + u^{3}/3 + u^{5}/5 + \cdots$), providing a smooth and invertible log-odds mapping compatible with integer scaling.

\smallskip
\noindent\textbf{Design Implications.}
These substitutions render the entire training process discrete and
reproducible.  
All gradient $g_i=p_i-y_i$, Hessian $h_i=p_i(1-p_i)$ for each data point $(\vecx_i,y_i)$, their (partial) sums $G,G_L,G_R,H,H_L,H_R$, and gain computations—previously
involving floating-point divisions—are now expressed as scaled integer ratios:
\begin{equation}
  \gain = \frac{1}{2}\!\left(
    \frac{G_L^2}{H_L+\lambda} +
    \frac{G_R^2}{H_R+\lambda} -
    \frac{G^2}{H+\lambda}
  \right) - \gamma,
  \label{eq:gainfp}
\end{equation}
where each division follows the fixed-point rule in~Eq.~\eqref{eq:fixed-muldiv}.  Additional numerical safeguards such as feature binning, probability clipping, and bounded leaf weights are applied as in standard XGBoost but are now part of the certified relation itself rather than ad-hoc implementation details.

Empirically (see Section~\ref{sec:experiments}), this fixed-point formulation reproduces floating-point XGBoost accuracy within statistical variation (less than 1\%), while enabling proof-friendly training semantics suitable for succinct zero-knowledge proofs.

\subsection{Numerical Safeguards in Fixed-Point XGBoost}
\label{app:sub:safeguards}
Practical XGBoost implementations incorporate a number of numerical safeguards. 
Here we make them explicit and they are enforced as part of the certified training process.

\noindent\textbf{Binary label validation.}
We enforce that each training label satisfies $y_i \in \{0,1\}$ by checking $y_i \cdot (1-y_i)=0$. This prevents malformed inputs that could invalidate gradient and Hessian computations.

\noindent\textbf{Clipping of probabilities and logits.}
When computing probabilities $p=\sigma(z)$ and their inverse logits $z=\log\frac{p}{1-p}$, we apply explicit clipping. Probabilities are restricted to $[p_{\min}, 1-p_{\min}]$ (with $p_{\min}$ typically around $10^{-6}$), and logits are bounded to $[-L,L]$ where $L=\log\frac{1-p_{\min}}{p_{\min}}$. These clamps serve two roles: (i) they prevent numerical instabilities such as division by zero or $\log(0)$, which would otherwise cause divergence between the implementation and the certificate, and (ii) they align the forward and inverse mappings so that round-trips $\textsc{SigmoidFP}(\textsc{LogitFromProbFP}(p))$ remain consistent under discretization. 

\noindent\textbf{Clipping of leaf weights.}
Each leaf weight is computed as $w=-G/(H+\lambda)$, but then passed through $\textsc{ClipFP}(w,-1,1)$. This a stability safeguard which avoids extreme updates when $H$ is very small and is also part of the certified relation: the verifier explicitly checks that the prover has applied the same clamp. Without it, it wouldn't be necessary that the certificate would match the training procedure.

\noindent\textbf{Pre-binning of features.}
Features are first mapped into $B$ equal-width bins, and the training algorithm works only on these integer bin indices. This removes floating-point comparisons on raw feature values from the relation and ensures that every branch decision in the tree is discrete making it easily verifiable in zero-knowledge.

\noindent\textbf{Division safety checks.}
For every fixed-point division $z = \lfloor x / y \rfloor$, the protocol enforces auxiliary range constraints on $x$, $y$, $z$, and the remainder to prevent wrap-around modulo the field. In particular, denominators and quotients are required to lie below explicit bounds so that $x = z \cdot y + r$ holds without overflow.

\noindent\textbf{Gain and histogram bounds.}
Aggregated gradient and Hessian values used in gain computation are range-checked to ensure that squaring and division operations do not overflow the field. These bounds are sufficient to guarantee correctness of split comparisons in the zero-knowledge proof.

\noindent\textbf{Deterministic handling of terminal nodes.}
When no split yields positive gain, a deterministic terminal split is enforced. This avoids ambiguity in tree structure and ensures that training and certification follow identical control flow.

\begin{procedure}[label={proc:Train}]{$\trainfp$}
\textbf{Parameters}: trees $m$, tree height $h$, learning rate $\eta$, regularizers $\lambda,\gamma$, \#bins $B$

\textbf{Input}: Training data $D=\{\vecx_i, y_i\}_{i=1}^n$.

\textbf{Output}: Sequence of trees $T_1,\ldots,T_m$ and base logit $z_0$.

\medskip
\underline{$\trainfp(\cD)$}:
\begin{algorithmic}[1]
    \State \textbf{for} $i=1,\ldots,n$: \textbf{assert} $y_i \in \{0,1\}$ 
		\State $p \gets (\sum_i y_i)/n$
    \State $\bar{p} \gets \clipfp(p, p_{\min}, p_{\max})$ 
    \State $z_0 \gets \logitfromprobfp(\bar{p})$ \Comment{$z_0 \approx \log(\bar{p} / (1 - \bar{p}))$}
		\State Let $ z_{0,i} = z_0$ for all $i\in[n]$
    \State $(\edges, \binID) \gets \prebinfp(\vecx, B)$
		\For{$k\in[m]$}
		\For{$i\in[n]$}
		\State $p_i \gets \sigmoidwidefp(z_{k-1,i})$ \Comment{$p_i \approx \frac{1}{1 + \exp(-z_{k-1,i})}$}
		\State $g_i \gets p_i - y_i$
		\State $h_i \gets p_i (1 - p_i)$
		\EndFor
		\State $T_k \gets (\vecf_k, \vectt_k, \vecw_k) = (\veczero, \veczero, \veczero)$
		\State $\buildtree(T_k,1,\{g_i\}_{i\in[n]}, \{h_i\}_{i\in[n]},\binID,\edges)$
		\For{$i\in[n]$}
			\State Let $l_{k,i}$ be the leaf node where $\vecx_i$ falls in $T_k$
			\State $z_{k,i} = z_{k-1,i} - \eta \cdot w_{k,l_{k,i}}$
		\EndFor
		\EndFor
    \State \Return{$\{T_k\}_{k\in[m]}$ and $z_0$}
\end{algorithmic}

\medskip
\underline{$\buildtree(T,\ell,\{g_i\}_{i\in \cI}, \{h_i\}_{i \in \cI},\binID,\edges)$}:
		\begin{algorithmic}[1]
		\If{$\text{height}(T) = h$}
			\State $G \gets \sum_{i\in \cI} g_i$; $H \gets \sum_{i\in \cI} h_i$
			\State $w' \gets \frac{G}{H + \lambda}$
			\State $w \gets \clipfp(w', -1, 1)$
			\State $T.w_{\ell-(N-1)} \gets w$
		\Else
			\State $(f^*, b^*,\gain^*) \gets \findsplit(\{g_i\}_{i\in \cI}, \{h_i\}_{i\in \cI}, \binID)$
			\If{$\gain^* \leq 0$} \Comment{pruning via dummy splits; any sample will go right}
      \State $b^* \gets \bdummy=0$  
      \State $f^* \gets \fdummy$ 
      \EndIf
				\State $t^* \gets \edges[f^*][b^*]$
				\State $\cL := \{i \in \cI : \binID[i][f^*] < b^*\}$; $\cR := \cI \setminus \cL$
				\State $\buildtree(T, 2\ell,\{g_i\}_{i\in \cL}, \{h_i\}_{i\in \cL}, \binID,\edges)$
				\State $\buildtree(T, 2\ell+1,\{g_i\}_{i\in \cR}, \{h_i\}_{i\in \cR}, \binID,\edges)$
				\State $(T.f_{\ell},T.t_{\ell}) \gets (f^*, t^*)$
			\EndIf
		\end{algorithmic}

    \medskip
		\underline{$\findsplit(\{g_i\}_{i\in \cI}, \{h_i\}_{i\in \cI}, \binID)$}:
		\begin{algorithmic}[1]
			\State $f^* \gets 0$; $b^* \gets 0$; $\gain^* \gets -\infty$
			\State $G \gets \sum_{i\in \cI} g_i$; $H \gets \sum_{i\in \cI} h_i$
			\For{$f\in[d]$}
        	\For{$b \in [B]$}
          	\State $\cL \gets \{i \in \cI : \binID[i][f] < b\}$
          	\State $G_{b,\cL} \gets \sum_{i\in \cL} g_i$; $H_{b,\cL} \gets \sum_{i\in \cL} h_i$
					\State $G_{b,\cR} \gets G - G_{b,\cL}$; $H_{b,\cR} \gets H - H_{b,\cL}$
                    \State 
					\State $\gain \gets \frac{1}{2}\cdot\left(\frac{G_{b,\cL}^2}{H_{b,\cL} + \lambda} + \frac{G_{b,\cR}^2}{H_{b,\cR} + \lambda} - \frac{G^2}{H + \lambda}\right) - \gamma$
					\If{$\gain^* < \gain$}
					  \State $\gain^* \gets \gain$; $f^* \gets f$; $b^* \gets b$
					\EndIf
				\EndFor
			\EndFor
			\State \Return{the best split $(f^*, b^*, \gain^*)$}
		\end{algorithmic}
	
\end{procedure}

	\section{Details of Certification Algorithm for XGBoost Training}\label{sec:xgboost-circuit}
In this section, we provide the detailed procedure for the certification that verifies the correctness of XGBoost training. The procedure is summarized in Procedure~\ref{proc:cert-xgboost}. 
The algorithm takes as input a training data $\vecx = (\vecx_1, \ldots, \vecx_n)$ with labels $\vecy = (y_1, \ldots, y_n)$, a fitted tree ensemble $\cT = \{T_k\}_{k\in[m]}$, where each tree $T_k$ is represented by $(\vecf_k, \vectt_k, \vecw_k)$, and a base logit $z_0$. It then verifies that $(\cT, z_0)$ is the output of the fixed-point XGBoost training algorithm $\trainfp$ (see Appendix~\ref{app:fixedpointxgb}) on input $(\vecx, \vecy)$.

\begin{claim}\label{claim:cert-xgboost}
	The certificate algorithm $\certxgb$ verifies the correctness of XGBoost training.
	That is, for any $(\vecx,\vecy), \cT = \{T_k\}_{k\in[m]}, z_0$, $\certxgb(\vecx,\vecy,\cT,z_0) = 1$ if and only if $(\cT,z_0)=\trainfp(\vecx,\vecy)$.
\end{claim}
\begin{proof}
	\textbf{(If)} If $(\cT,z_0)=\trainfp(\vecx,\vecy)$, then all the assertions in the subprotocols clearly hold, and thus $\certxgb$ returns 1.
	
	\textbf{(Only if)} 
  For fixed $(\vecx,\vecy,\cT,z_0)$, suppose $(\cT=(\vecf_k,\vectt_k,\vecw_k)_{k=1}^m, z_0) \neq (\cT'=(\vecf'_k,\vectt'_k,\vecw'_k)_{k=1}^m,z'_0)=\trainfp(\vecx,\vecy)$. 
  We show that $\certxgb(\vecx,\vecy,\cT,z_0) = 0$. 

  \underline{If $z_0 \neq z'_0$:} Since $\validatelogit$ computes the logit from $\vecy$ as in $\trainfp$, $\certxgb$ derives the same $z'_0$ from $\vecy$. Thus, it holds that $z_0 \neq z'_0$, and thus $\certxgb$ returns 0.
  
  \underline{If $T_1\neq T'_1$ and $(\vecf_1,\vectt_1)\neq (\vecf'_1,\vectt'_1)$ (i.e., the first tree contains an incorrect split):} Since $\certxgb$ runs $\inithists$, which computes $p_i$, $g_i$, and $h_i$ from $z_{i,0} = z_0$ for $i=1,\ldots,n$ as in $\trainfp$, they are the same as those in $\trainfp$. 
  $\inithists$ also aggregates them into root histograms from leaf histograms. 
  Thus, every possible split and $\gain$ computed inside $\validatesplits$ for the root are the same as those in $\findsplit$ of $\trainfp$. 
  
  We first consider the case where $(f_{1,1},t_{1,1}) \neq (f_{1,1}',t_{1,1}')$, i.e., the split at the root node is incorrect.
  In $\certxgb$, it computes $(f^*_{1,1},t^*_{1,1})$ and the corresponding max $\gain$ from the root histograms, which are the same as those in $\trainfp$.
  It runs $\validatesplits$, which checks either of the following cases: (1) if $\gain^*$ is non-positive, then it checks whether $(f_{1,1},t_{1,1})$ has dummies; (2) else, it checks whether $(f_{1,1},t_{1,1}) = (f^*_{1,1}, t^*_{1,1})$. 
  In the former case, $\trainfp$ would also set dummies for $(f'_{1,1},t'_{1,1})$ and thus the assertion by $\validatesplits$ fails.
  In the latter case, since $\trainfp$ chooses the optimal split $(f^*_{1,1},t^*_{1,1})$  which is different from $(f_{1,1},t_{1,1})$, the assertion by $\validatesplits$ also fails.
  Note that $\validatesplits$ also checks that $(f^*_{1,1},b^*_{1,1})$ is the smallest among those achieving the max $\gain$, which also holds since $\trainfp$ chooses the smallest one in case of ties.

  Now, we consider the case where $(f_{1,1},t_{1,1}) = (f_{1,1}',t_{1,1}')$, i.e., the split at the root node is correct, but some other split in $T_1$ is incorrect.
  By $\validateinference$, the samples reaching left and right child nodes of the root in $\certxgb$ are the same as those in $\trainfp$.
  Thus, the histograms at these child nodes computed by $\inithists$ are also the same as those in $\trainfp$.
  Then, we can apply the same argument as above to either of the child nodes where the split is incorrect, and show that assertion by $\validatesplits$ fails in one of the second-level nodes. 
  Applying this argument recursively, $\certxgb$ always detects the incorrect split in $T_1$.

  \underline{If $T_1\neq T'_1$ and $(\vecf_1,\vectt_1) = (\vecf'_1,\vectt'_1)$ but $\vecw_1 \neq \vecw'_1$ (i.e., the first tree contains an incorrect leaf weight):} Since the splits are correct, by $\validateinference$, the samples reaching each leaf in $\certxgb$ are the same as those in $\trainfp$.
  Thus, the histograms at these leaves computed by $\inithists$ are also the same as those in $\trainfp$.
  Then, $\validateleafweights$ checks whether the leaf weights are correctly computed from these histograms as in $\trainfp$, and thus the assertion fails for the incorrect leaf weight.

  \underline{$T_1 = T'_1$ but for some $k\leq m$, $T_k \neq T'_k$ (i.e., one of the subsequent trees is incorrect):} 
  Since $\certxgb$ correctly computes the updated logits $\vecz_1$ after training $T_1$, which are the same as those in $\trainfp$, we can apply the same arguments as above iteratively to $T_k$ for $k=2,\ldots,m$ and show that $\certxgb$ detects the incorrect tree.

	
	Therefore, by induction, for each $k\in[m]$, $T_k$ is correctly trained on $(\vecx,\vecy)$ with the initial scores $\vecz_{k-1}$, and the updated scores $\vecz_k$ are correctly computed. This implies that $(\cT,z_0)=\trainfp(\vecx,\vecy)$.
\end{proof}

\begin{procedure}[label={proc:cert-xgboost}]{$\certxgb$}
\textbf{Parameters}: Number of points $n$, features $d$, trees $m$, bins $B$, leaves $N=2^h$ at depth $h$, learning rate $\eta$, regularizers $\lambda,\gamma$.

\textbf{Input}: Fixed-point training data $\cD = (\vecx, \vecy)$, the fitted tree ensemble $\cT = \{T_k\}_{k\in[m]}$, and base logit $z_0$.

\textbf{Output}: 1 (accept) or 0 (reject).
 \begin{algorithmic}[1]
  \State \textbf{for} $i=1,\ldots,n$: \textbf{assert} $y_i \in \{0,1\}$
	\State $(\edges,\binID)\gets\prebinfp(\vecx,B)$
	\State \textbf{assert} $\validatelogit(\vecy,z_0)=1$
	\State $\vecz_0 \gets (z_{0,i})_{i=1}^n$ with $z_{0,i} = z_0$
    \For{each tree $k\in[m]$} \Comment{Initialize intermediate scores}
    	\For{each data point $i\in[n]$}
		\State compute the reached leaf $l_{k,i}\in[N]$ by evaluating $T_k$ on $\vecx_i$
		\State $z_{k,i} \gets z_{k-1,i}-\eta\cdot w_{k,l_{k,i}}$
		\EndFor
	\State $\vecz_k \gets (z_{k,i})_{i=1}^n$
	\State $\vecl_k \gets (l_{k,i})_{i=1}^n$
    \EndFor
	\For{each tree $k\in[m]$} \Comment{Tree validation}
	\State \textbf{parse} $(\vecf_k,\vectt_k,\vecw_k)\gets T_k$
	\State \textbf{assert} $\validateinference(\vecx,\vecf_k,\vectt_k,\vecl_k)=1$
	\State $(G_k,H_k)\gets\inithists(\vecz_{k-1},\vecy,\binID,\vecl_k)$
  \For{each internal node $\ell\in[N-1]$}
  \State compute $(f^*_{k,\ell},b^*_{k,\ell})$ such that $\gain$ derived from $(G_k,H_k)$ is maximized
  \State $t^*_{k,\ell}\gets\edges[f^*_{k,\ell}][b^*_{k,\ell}]$
  \EndFor
  \State $(\vecf^*_k,\vectt^*_k) \gets ((f^*_{k,\ell}, t^*_{k,\ell}))_{\ell=1}^{N-1}$
  \State \textbf{assert} $\validateleafweights(G_k,H_k,\vecw_k)=1$
  \State \textbf{assert} $\validatesplits(G_k,H_k,\vecf_k,\vectt_k,\vecf^*_k,\vectt^*_k,\edges)=1$
	\EndFor
	\State \Return{1 (accept)}
 \end{algorithmic}
\end{procedure}

\begin{procedure}[label={prot:prebin-fp}]{$\prebinfp$}
\textbf{Input:} Fixed-point feature matrix $\vecx=(x_{i,j})_{i\in[n],\,j\in[d]}$, number of bins $B$\\
\textbf{Output:} $\edges\in \NN^{d} \times \RR^{B+1}$, and $\binID\in[B]^{n\times d}$
\begin{algorithmic}[1]
\For{$f=1$ \textbf{to} $d$}
  \State $c_{min} \gets \min_{i\in[n]} x_{i,f}$;\quad $c_{max} \gets \max_{i\in[n]} x_{i,f}$
    \State $\delta_f \gets  \dfrac{c_{max}-c_{min}}{B}$
    \State Define a lookup table $\edges[f] \gets (c_{min} + (b-1)\cdot \delta_f)_{b=1}^B$ consisting of left edges of $B$ equal-width bins.
    \State $\edges[\fdummy][\bdummy] \gets \tdummy$ \Comment{define dummy (small enough) edge for pruning, e.g., $\tdummy = \mathtt{DBL\_MIN}$}
    \State \textbf{for } $i=1,\dots,n$:
      \[
        \binID[i][f] \gets \sum_{b=1}^B \bm{1}\{\edges[f][b] \le x_{i,f} \}
      \]
\EndFor
\State \Return $(\edges,\binID)$
\end{algorithmic}
\end{procedure}

\begin{procedure}[label={prot:validate-logit}]{$\validatelogit$}
\textbf{Parameters:} $p_{\min}, p_{\max}$ for clipping probabilities\\
\textbf{Input:} Labels $\vecy=(y_i)_{i\in[n]}$ and base logit $z_0$  \\
\textbf{Output:} Accept/Reject
\begin{algorithmic}[1]
\State $p \gets \frac{1}{n}\sum_{i=1}^n y_i$
\State $p' \gets \clipfp(p, p_{\min}, p_{\max})$
\State $u\gets 2p'-1$ 
\State $z'_0 \gets 2\cdot\Big(u + \frac{u^3}{3} + \frac{u^5}{5}\Big)$ \Comment{truncated $\mathrm{atanh}$ series with fixed-point ops}
\State \textbf{assert} $z_0 = z'_0$ 
\end{algorithmic}
\end{procedure}

\begin{procedure}[label={prot:validate-inference}]{$\validateinference$}
\textbf{Input:} Data points $\vecx = (x_{i,f})_{i\in[n],f\in[d]}$, node features $\vecf$, node thresholds $\vectt$, and leaf indices $\vecl = (l_{i})_{i=1}^n$ reached by each $\vecx_i$\\
\textbf{Output:} Accept/Reject

\begin{algorithmic}[1]
\State Define a table of root-to-leaf paths $\treepath[l] = (f_{l,j}, t_{l,j})_{j=1}^{h}$ for each $l\in[N]$
\For{$i=1$ \textbf{to} $n$}
  \State \emph{(Path lookup)} Retrieve $\treepath[l_i] = (f_{l_i,j}, t_{l_i,j})_{j=1}^{h}$.
  \State \emph{(Batched permutation)} Define a permutation $(\bar{x}_{i,1},\dots,\bar{x}_{i,d})$ of $(x_{i,1},\dots,x_{i,d})$ such that the first $h$ entries align with the path features $(f_{l_i,1},\dots,f_{l_i,h})$; \textbf{assert} permutation correctness.
  \State \emph{(Branch decisions)} For $j=1,\dots,h$, set $a_j \gets \bm{1}\{\,t_{l_i ,j} \leq \bar{x}_{i,j}\,\}$ and \textbf{assert}:
  \[
    l_{i} + 2^h - 1 \;=\; 2^h + \sum_{j=1}^{h} a_j\cdot 2^{h-j}.
  \]
\EndFor
\end{algorithmic}
\end{procedure}

\begin{procedure}[label={prot:init-hists}]{$\inithists$}
\textbf{Input:} Scores $\vecz = (z_i)_{i=1}^n$, labels $\vecy=(y_i)_{i=1}^n$, $\binID$, leaf indices $\vecl=(l_i)_{i=1}^n$\\
\textbf{Output:} Histograms $(G,H)$
\begin{algorithmic}[1]
\State \textbf{for} $f\in[d],\,\ell\in[2N-1],\,b\in[B]$, initialize $G[f][\ell][b]\gets 0$, $H[f][\ell][b]\gets 0$.
\For{$i=1$ \textbf{to} $n$}
  \State $p_i \gets \sigmoidwidefp(z_i)$;\quad $g_i \gets p_i - y_i$;\quad $h_i \gets p_i\cdot(1-p_i)$
  \For{$f=1$ \textbf{to} $d$}
    \State $b\gets \binID[i][f]$;\; $\ell\gets l_i +(N-1)$
    \State $G[f][\ell][b] \gets G[f][\ell][b] + g_i$;\quad $H[f][\ell][b] \gets H[f][\ell][b] + h_i$
  \EndFor
\EndFor
\For{$\ell=N-1$ \textbf{down to} $1$} \Comment{propagate up to root}
  \For{$f=1$ \textbf{to} $d$}
    \For{$b=1$ \textbf{to} $B$}
      \State $G[f][\ell][b] \gets G[f][2\ell][b] + G[f][2\ell+1][b]$
      \State $H[f][\ell][b] \gets H[f][2\ell][b] + H[f][2\ell+1][b]$
    \EndFor
  \EndFor
\EndFor
\State \Return $(G,H)$
\end{algorithmic}
\end{procedure}

\begin{procedure}[label={prot:sigmoid-wide-fp}]{\sigmoidwidefp}
\textbf{Input:} Fixed-point logit\\
\textbf{Output:} $y \in \{0,\dots,\scale\}$
\begin{algorithmic}[1]
\If{$z \le -2$}
  \State \Return $0$
\ElsIf{$z \ge 2$}
  \State \Return $1$
\Else
  \State $y \gets \dfrac{z + 2}{4}$ \Comment{integer arithmetic only}
  \State \Return $y$
\EndIf
\end{algorithmic}
\end{procedure}

\begin{procedure}[label={prot:validate-leaf-weights}]{$\validateleafweights$}
\textbf{Input:} Leaf histograms $(G,H)$, weights $\vecw=(w_l)_{l=1}^N$\\
\textbf{Output:} Accept/Reject
\begin{algorithmic}[1]
\For{$l=1$ \textbf{to} $N$}
  \State $\ell \gets l + (N-1)$
  \State $G^{\mathrm{leaf}} \gets \sum_{b=1}^B G[1][\ell][b]$;\quad $H^{\mathrm{leaf}} \gets \sum_{b=1}^B H[1][\ell][b]$ \Comment{use any feature $f$, as the sum is feature-independent}
  \State $w'_l \gets \dfrac{G^{\mathrm{leaf}}}{H^{\mathrm{leaf}}+\lambda}$
  \State \textbf{assert} $w_{l} = \clipfp(w'_l,-1,1)$
\EndFor
\end{algorithmic}
\end{procedure}

\begin{procedure}[label={prot:validate-splits}]{$\validatesplits$}
\textbf{Input:} Node histograms $(G,H)$, resulting splits $(\vecf,\vectt)=(f_\ell,t_\ell)_{\ell=1}^{N-1}$, indices $(\vecf^*,\vectt^*)=(f^*_\ell,t^*_\ell)_{\ell=1}^{N-1}$ leading to maximum gain,  and $\edges$. Let $e_0 = 0$. \\
\textbf{Output:} Accept/Reject
\begin{algorithmic}[1]
\For{$\ell=1$ \textbf{to} $N-1$}
  \For{$f=1$ \textbf{to} $d$}
    \State $G\gets\sum_{b=1}^B G[f][\ell][b]$;\quad $H\gets\sum_{b=1}^B H[f][\ell][b]$
    \State $G_L\gets 0,\,H_L\gets 0$
    \For{$b=1$ \textbf{to} $B$}
      \State $G_L\gets G_L + G[f][\ell][b]$;\quad $H_L\gets H_L + H[f][\ell][b]$
      \State $G_R\gets G - G_L$;\quad $H_R\gets H - H_L$
      \State $\gain[f][b] \gets \frac{1}{2}\cdot\left(\frac{G_L^2}{H_L+\lambda} + \frac{G_R^2}{H_R+\lambda} - \frac{G^2}{H+\lambda}\right)-\gamma$ \Comment{fixed-point divisions}
    \EndFor
  \EndFor
  \State Retrieve $b^*_\ell$ s.t.\ $t^*_\ell = \edges[f^*_\ell][b^*_\ell]$
  \State Retrieve $\gain^* \gets \gain[f^*_\ell][b^*_\ell]$
  \State \textbf{assert} $(\gain^* > \gain[f][b]) \lor (\gain^* = \gain[f][b] \land f^*_\ell \leq f \land b^*_\ell \leq b)$ for all $f\in[d],b\in[B]$ \Comment{tie-breaker check} 
  \State $e_{\ell} \gets (\gain^* \leq 0 \lor e_{parent})$ \Comment{check if the current node should be terminal, or the upper-level node is terminal}
  \State \textbf{assert} $(1-e_{\ell})\cdot (f_\ell = f^*_\ell  \land t_\ell = t^*_\ell) \lor e_\ell \cdot (f_\ell = \fdummy \land  t_\ell = \tdummy) = 1$ \Comment{if not dummy, use max gain split; else, use terminal split}
\EndFor
\end{algorithmic}
\end{procedure}

\begin{procedure}[label={prot:ClipFP}]{$\clipfp$}
\textbf{Input:} $x$, lower bound $a$, upper bound $b$ with $a \le b$ represented as fixed-point numbers\\
\textbf{Output: }Clipped value $y$.
\begin{algorithmic}[1]
  \State $y \gets x$
  \If{$y < a$} \State $y \gets a$ \EndIf
  \If{$y > b$} \State $y \gets b$ \EndIf
  \State \Return $y$
\end{algorithmic}
\end{procedure}

    \section{Additional Details on ZK XGBoost-FP}
\label{sec:app:zkxgboostfp}

In this section, we provide detailed explanation on the proof of comparison, division, and truncation, as well as additional components in proving XGBoost-FP shown in Figure~\ref{proc:cert-xgboost}.

\subsection{Proof of Comparison}
\label{sec:app:comparison}
Our proof of comparison builds upon the idea of~\cite{hao2024scalable} to prove the comparison relation by bits-decomposition. It first converts the proof of comparison into a MSB proof $\bm{1}\{\auth{x} < \auth{y}\} = {\sf MSB}(\auth{x - y})$ due to the signed representation. This is true because ${\sf MSB}(\auth{x - y})=1$ if $x-y<0$ and ${\sf MSB}(\auth{x - y})=0$ otherwise. We shift our focus into proving $\auth{s} = {\sf MSB}(\auth{z})$ for committed $(\auth{z},\auth{s})$. To ensure the soundness of MSB proof, it requires a \emph{one-bit gap} between the scale of encoded signed number and the underlying field: for a finite field $\FF_p$, we require $x,y \in [0,\lfloor p/4\rfloor) \cup [p - \lfloor p/4 \rfloor,p-1]$. Otherwise, the subtraction may cause an overflow and lead to an incorrect result. 
Existing approaches all somehow require the $\cP$ providing the rest bits of $z$ and proving the correctness of bit-decomposition.

\noindent\textbf{Bit-decomposition.}
Assume that a $\cP$ commits to $(\auth{z},\auth{s})$ and needs to prove that $\auth{s} = {\sf MSB}(\auth{z})$. Define $n,\FF_p$ s.t. $\log|\FF_p| \leq n$, a naive way to prove the sign is to apply a bit-decomposition as in~\cite{USENIX:WYXKW21}. However, it incurs $O(n)$ communication cost to commit to all $n$ bits $(z_0,\dots,z_{n-1})$ and need to prove an additional binary adder circuit. 
\cite{hao2024scalable}'s approach only incurs $O(n/d)$ for an arbitrarily defined $d$ by leveraging the proof of table lookup~\cite{CCS:FKLOWW21,yang2024two,DBLP:journals/iacr/Habock22a}. 

Without loss of generality, we assume $(d,t)$ such that $n - 1 = d \cdot t$. The idea is to group every consecutive $d$ bits of $z$ into $(\tilde{z}_0,\dots,\tilde{z}_{t-1})$, in which $\tilde{z}_i = \sum_{j=0}^t 2^j \cdot z_{id+j}$. In this case, $\cP$ only needs to commit to $t = (n-1) / d + 1$ elements and prove that $z = 2^{n-1} \cdot s + \sum_{i=0}^{t-1} 2^{id} \cdot \tilde{z}_i$. Additionally, $\cP$ needs to show that each $\tilde{z}_i \in [0,2^d-1]$ via a membership proof from a list $T = (0,1,\dots,2^d-1)$, which can be instantiated by table lookup with only $O(1)$ cost per access~\cite{yang2024two}.

\noindent\textbf{Preventing Malicious Prover.} A soundness issue with the above scheme is that a cheating $\cP$ may commit to incorrect $(\tilde{z}_0,\dots,\tilde{z}_{t-1}, s)$ such that
$z + p = 2^{n-1} \cdot s + \sum_{i=0}^{t-1} 2^{id} \cdot \tilde{z}_i$. 
For example, assume that $p = 2^{61}-1$ and $z = 0$, a cheating $\cP$ can commit to $s = 1$ and all $\tilde{z}_i = 2^d-1$, and prove that $z = p \; {\sf mod} \; p = 0$.
Since ${\sf mod} \; p$ happens implicitly in ZK operations, extra measures are needed to detect such wrap-around.~\cite{hao2024scalable} attempted to address this issue but its approach doubles the cost of secure comparison.

Our solution employs a Mersenne prime in the form of $p = 2^n - 1$, and rules out the only cheating case when $z=0$, $s=1$ and $\tilde{z}_i = 2^d-1$ for all $i \in [0,t)$. 
By defining $\auth{w} = \bm{1}\{\auth{z} \neq 0\}$, we observe that 
$
{\sf MSB}(z) = \begin{cases}
   0 &\text{if } z = 0, w = 0 \\
   s &\text{if } z \neq 0, w = 1
\end{cases}
$.
It implies that $(1 - w) \cdot s = 0$ should always be true. Additionally, we employ the non-equality-zero check from~\cite{SP:PHGR13} to prove $\auth{w} = \bm{1}\{\auth{z} \neq 0\}$. Our solution to address the soundness issue only requires committing 2 extra values and proving 3 multiplicative relations, which incurs less overhead than~\cite{hao2024scalable}. We defer the formal algorithm description to Figure~\ref{prot:zkp-nonlinear} in Appendix~\ref{sec:app:zkxgboostfp}.


\noindent\textbf{Cost Analysis.} The proof of comparison requires committing to $3+n/d$ witnesses, performing $n/d$ lookup on a table of size $2^d$, and proving $3$ multiplicative relations. The proof for table lookup generally takes $O(1)$ per access. Its $O(2^d)$ setup cost will be amortized since the table is reused across the whole proof~\cite{yang2024two,DBLP:journals/iacr/Habock22a}.

\begin{protocol}[label={prot:zkp-nonlinear}]{ZKP Gadgets}

\noindent\textbf{Check Nonzero.} On input $\auth{x}, \auth{y}$, prove that $\auth{y} = \bm{1}\{\auth{x} \neq 0\}$.
\begin{enumerate}
    \item If $x = 0$, $\cP$ commits to $\auth{u}$ such that $u = 0$. Otherwise, it commits to $u = x^{-1}$.
    \item Prove that $\auth{y} - \auth{u} \cdot \auth{x} = 0$ and $(1 - \auth{y}) \cdot \auth{x} = 0$.
\end{enumerate}

\medskip
\noindent\textbf{Comparison.} On input $\auth{x}, \auth{y}, \auth{s}$, prove that $s = \bm{1}\{\auth{x} < \auth{y}\}$. Define parameters $(d,t)$ such that $n - 1 = d \cdot t$. Construct a public lookup table $\mathcal{T} = (0,\dots,2^d-1)$.
\begin{enumerate}
    \item Commits to $\auth{z} = \auth{x} - \auth{y}$ and $\auth{s} = {\sf MSB}(\auth{z})$. Decompose $z$ into $(\auth{\tilde{z}_0},\dots,\auth{\tilde{z}_{t-1}})$ such that $\tilde{z}_i = \sum_{j=0}^t 2^j \cdot z_{id+j}$.
    \item Prove that for $i \in[0,t)$, $\tilde{z}_i = \mathcal{T}[\tilde{z}_i]$. 
    \item Prove that $\auth{z} = 2^{n-1} \cdot \auth{s} + \sum_{i=0}^{t-1} 2^{id} \cdot \auth{\tilde{z}_i}$.
    \item Commit to $\auth{w} = \bm{1}\{\auth{z} \neq 0\}$ and prove its correctness using the above nonzero check. Prove that $(1 - \auth{w}) \cdot \auth{s} = 0$.
\end{enumerate}

\medskip
\noindent\textbf{Division.} On input $\auth{x}, \auth{y}, \auth{z}$, prove that $z = \lfloor x / y \rfloor$. Define $m_{q}$ and $m_d$ as the upper bound of the abstract value of the quotient and denominator.
\begin{enumerate}
    \item $\cP$ commits to abstract values $\auth{\bar{x}}, \auth{\bar{y}}, \auth{\bar{z}}$ and prove that $\bar{i} = (1 - 2 \cdot {\sf MSB}(i)) \cdot i$ for $i \in \{x,y,z\}$.
    \item $\cP$ commits to the residual $\auth{\bar{r}}$ and prove that $\auth{\bar{x}} = \auth{\bar{y}} \cdot \auth{\bar{z}} + \auth{\bar{r}}$.
    \item Prove that $0 \leq \auth{\bar{r}} < \auth{\bar{y}}$, $\auth{\bar{y}} \leq m_d$ and $\auth{\bar{z}} < m_q$.
    \item Prove that ${\sf MSB}(z) = {\sf MSB}(x) \oplus {\sf MSB}(y)$ using the fact that for any binary values $\alpha,\beta$, $\alpha \oplus \beta = \alpha + \beta - 2 \alpha \cdot \beta$.
\end{enumerate}

\medskip
\noindent\textbf{Truncation.} On input $\auth{x}, \auth{z}, f$, prove that $z = \lfloor x / 2^f \rfloor$. The proof is the same as the division except that $y=2^f$ is public so that its range check is avoided. Additionally, $z$ is bounded by $m_q = \lfloor p/2^{f+1} \rfloor$. If $p$ is a Mersenne prime, the check of $0 \leq \auth{r} < 2^f$ and $z < m^q$ can both be done by bit-decomposition proofs.

\end{protocol}

\subsection{Proof of Division and Truncation}

\textbf{Division.} Proof of division takes inputs $(\auth{x}, \auth{y}, \auth{z})$ and proves that $z = \lfloor x / y \rfloor$. Note that this statement disallows directly proving it by $x = z \cdot y$ because of the floor operation. To simplify the problem, we assume that both $(x,y)$ represent positive values. To generalize it to arbitrary inputs, we can first prove their absolute values and then determine the sign of $z$ based on the signs of $x,y$. It can also be trivially generalized for fixed-point representations by left-shift $x$ before the division (assume that the left shift does not overflow).

In our XGBoost-FP, the numerators and denominators may range from $0$ to nearly $p/2^f$. Previous proofs of division are not suitable since they only work for bounded small values~\cite{CCS:LiuXieZha21,DBLP:conf/ccs/Pappas024}. Another approach proves the floating-point division but requires thousands of constraints~\cite{weng2021mystique}.

To prove the relation, our approach leverages the existence of a residue $r \in [0,y)$ such that $x = z\cdot y + r$. Additionally, two range checks are required to maintain the soundness: $r \in [0,y)$, and $z \in [0,\lfloor p/y \rfloor]$. The first check ensures $r$ is a residue. The second check is needed to prevent the similar wrap-around issue happened in the proof of MSB. Namely, there could be many possible pairs of $(z', r')$ such that $x = z' \cdot y + r' \; {\sf mod} \; p$ is $z'$. The proof should ensure a $z$ that satisfies $x = z \cdot y + r$ without modulo $p$.

In our work, we define upper bounds for the positive denominator $y$ and quotient $z$, denoted as $m_d$ and $m_q$, such that $m_d \cdot m_q < p/2$. $\cP$ proves that $\auth{y} < m_d$, $\auth{r} < \auth{y}$ and $\auth{z} < m_q$. This is not necessarily sufficient for general proof of division, but is sufficient for our XGBoost-FP since with a relatively large $p$, these values are expected to lie within a reasonable range.

\noindent\textbf{Truncation.} The truncation is a special type of division with public and positive denominators $y = 2^f$. It follows similar ideas as in the above proof of division to prove the range of $r$ and $z$. However, its proof is much simpler because $f$ is usually small so that the range check of $r \in [0,2^f)$ can be implemented by a table lookup. Also, the Mersenne prime $p = 2^n-1$ results in $\lfloor (p-1) / 2^f \rfloor = 2^{n-f} - 1$, which also enables a range check of $z \in [0, 2^{n-f} - 1]$ to be resolved by our bit-decomposition proof.

We defer the formal algorithm description for division and truncation to Figure~\ref{prot:zkp-nonlinear} in Appendix~\ref{sec:app:zkxgboostfp}.

\noindent\textbf{Cost Analysis.} The division proof invokes $7$ comparison/MSB and proves $4$ multiplicative relation. The truncation is simplified to $2$ MSB proof, 1 bit-decomposition proof, 1 table lookup, and 2 multiplications. The concrete overhead depends on the underlying proof system and the construction of lookup table.

\subsection{Additional Components}
\label{sec:app:zkp:additional}
\noindent\textbf{Input validation.} Our model training requires the global truth $\bm{y}$ to be binary values. Hence, a check is performed on witnesses to ensure that $y_i \cdot (1 - y_i) = 0$ for $i \in [n]$. Our protocol also supports arbitrary input validation on input dataset $\bm{x}$ depending on the dataset type.

\noindent\textbf{Global Range Checks.} In principle, the proof of numerical computation loses soundness whenever the overflow occurs. This implies that a range check is needed for every intermediate values that are committed during the proof. This can be realized by the proof of comparison mentioned above. To reduce the overhead for range checks, we examine the XGBoost proof and find out that most of intermediate values are already bounded because of the checks that we applied in the proof of comparison, division, and truncation. Additional range checks are only required in a few places.

Essentially, the range check is required when validating the leaf weights and the splits of internal nodes. The former computes $w \leftarrow G / (H + \lambda)$. To respect the fixed-point representation, the proof in fact validates $(G \cdot {\sf scale}) / (H + \lambda)$ so that the result $w$ is lifted by ${\sf scale}$. Hence, the value $|G|$ should be restricted in $[0, p / (2\cdot {\sf scale}))$. Similarly, we compute $G^2 / (H + \lambda)$ when validating the split. Since $G^2$ is already lifted by ${\sf scale}^2$, it can be directly fed into the proof of division without a truncation. However, we still need to validate that $G \in [-\sqrt{p/2}, \sqrt{p/2})$ to prevent the overflow when proving $G^2$.

In some situations, the checks can be avoided if the input to specific operations are already bounded. When validating base logit $z_0'$, the approximation of atanh is guaranteed not to overflow since $p \leftarrow \sum_{i=1}^n y_i / n < 1$, due to $y_i$ being binary values. In validating the hessian $h_i \leftarrow p_i \cdot (1 - p_i)$, the overflow is also not needed because $p_i$ is an output of the sigmoid function.


\noindent\textbf{Public Commitment and Proof of Opening.} Our framework also supports the commit-and-prove ZKP such that $\cP$ can first publicly commit to the dataset $\mathcal{D} = (\bm{x}, \bm{y})$, model parameters $T_k$, and other auxiliary information, and later prove their consistency in ZKP. The public commitment can be transmitted to a public bulletin board or corresponding verifiers anytime before ZKP. This can be realized by~\cite{USENIX:WYXKW21,sun2025committed,campanelli2019legosnark}.

\subsection{Security Analysis}
\label{sec:app:zkp:security}

We provide the security analysis for Theorem 1. We assume that $\Func{ZK}$ is instantiated by~\cite{yang2024two,DBLP:journals/iacr/Habock22a} and $\Func{CVOLE}$ is instantiated by~\cite{sun2025committed}. The secure realization of these functionalities can be referred to the cited works. We show Functionalities $\Func{CVOLE}$ and $\Func{ZK}$ below and reinstate Theorem 1.

\begin{functionality}[label={func:cvole}]{$\Func{CVOLE}$}
The functionality interacts with three parties: a prover $\cP$, a verifier $\cV$, and an ideal adversary $\cS$. It is parameterized by a secret $\bm{x}$ and its commitment ${\sf com}_{x}$.

\begin{enumerate}
    \item Upon receiving ${\sf init}$ from both parties, and additional input $\bm{x}$ and decommitment message ${\sf decom}_{x}$ from $\cP$, check whether ${\sf com}_{x}$ correctly decommits to $\bm{x}$.
    \item Upon receiving ${\sf cvole}$ from both parties, sample a VOLE correlation $\bm{m} = \bm{k} + \bm{x} \cdot \Delta$. Send $(\bm{m})$ to $\cP$ and $(\bm{k},\Delta)$ to $\cV$. If $\cP$ is corrupted, receive $(\bm{m})$ from $\cS$ and recompute $\bm{k} = \bm{m} - \bm{x} \cdot \Delta$. If $\cV$ is corrupted, receive $(\bm{k},\Delta)$ from $\cV$ and recompute $\bm{m} = \bm{k} + \bm{x} \cdot \Delta$. Send these values to corresponding parties.
\end{enumerate}

\end{functionality}

\begin{functionality}[label={func:zk}]{$\Func{ZK}$}
The functionality interacts with three parties: a prover $\cP$, a verifier $\cV$, and an ideal adversary $\cS$. It takes input a circuit $C$ that consists of arithmetic components and read-only or read-write RAM components. Upon receiving $({\sf prove},C,\bm{x},\bm{w})$ from $\cP$ and $({\sf verify},C,\bm{x})$, it verifies $C(\bm{x},\bm{w})$ by traversing the circuit by topological order and check the relations.

\begin{enumerate}
    \item For arithmetic components, it verifies arithmetic relations. If any relation is not satisfied by $\bm{x},\bm{w}$, it aborts.
    \item For RAM components, it initializes the memory using data from $\bm{x}$ or $\bm{w}$. It checks the data output by ${\sf read}$ accesses, and modifies the table according to ${\sf write}$ accesses. If any access operation is inconsistent, it aborts.
\end{enumerate}

\end{functionality}

\noindent\textbf{Theorem 1.} \textit{Define the relation by the proof of fixed-point XGBoost Training described in $\relxgb$~\eqref{eq:rel},  Protocol \ref{prot:zkp-training-short} securely realizes $\Fcp$ in the $(\Func{CVOLE},\Func{ZK})$-hybrid model.}

\textbf{Completeness.} In the case of a pair of honest $\cP$ and $\cP$, the completeness is obvious since $\cP$ honestly trains the model from the committed dataset and commits to the extended witness. Assume that the prover trains the model and proves ZK XGBoost both using fixed-point computation, the proof passes with probability $1$.

\textbf{Soundness.} Our protocol makes black-box access to the underlying $\Func{ZK}$ and $\Func{CVOLE}$. Hence, a corrupted $\cP$'s view can be simulated by constructing a simulator $\simulator$ who emulates the functionalities $\Func{ZK}$ and $\Func{CVOLE}$. 
In $\Func{CVOLE}$, it receives $\cD = (\vecx,\vecy)$, $\cT=\{T_k\}_{k\in[m]}$ with $T_k=(\vecf_k,\vectt_k,\vecw_k)$, $z$, and their commitments from $\cP$ and relay them to $\Fcp$. 
It also check whether the commitments successfully decommits to these values. If $\Func{CVOLE}$ aborts, it sends ${\sf abort}$ to $\Fcp$ and aborts itself.
It receives IT-MACs from corrupted $\cP$ and samples a correct VOLE correlation for committed data.
It emulates $\Func{ZK}$ with $\cP$'s input and the XGBoost-FP relation, and $\Func{ZK}$ checks whether the XGBoost-FP relation is satisfied. If not, it sends ${\sf abort}$ to $\Fcp$ and aborts itself. Otherwise, it sends ${\sf PROVE-RECEIPT}$.

\textbf{Zero-Knowledge.} $\cV$ has no input to $\Fcp$ so the emulation of $\cV$'s view is relatively simple. We construct a simulator $\simulator$, who emulates $\Func{ZK}$ and $\Func{CVOLE}$, and interact with corrupted $\cV$ accordingly. It relays whatever it receives from $\Fcp$ to corrupted $\cV$ and aborts when it aborts.
	\section{Details of Certificate for Random Forest Training}\label{sec:forest-circuit}
In this section, we provide the detailed procedure for the certificate that verifies the correctness of random forest training. The procedure is summarized in Procedure~\ref{proc:cert-forest}. 
The algorithm takes as input a training data $\vecx = (\vecx_1, \ldots, \vecx_n)$ with labels $\vecy = (y_1, \ldots, y_n)$, a fitted forest $\cT = \{T_k\}_{k\in[m]}$, where each tree $T_k$ is represented by $(\vecf_k, \vectt_k, \vecw_k)$.

\begin{procedure}[label={proc:cert-forest}]{$\certforest$}
\textbf{Parameters}: Number of points $n$, features $d$, trees $m$, bins $B$, leaves $N=2^h$ at depth $h$.

\textbf{Input}: Fixed-point training data $(\vecx, \vecy)$, the fitted forest $\cT = \{T_k\}_{k\in[m]}$, index sets $\{I_k\}_{k\in[m]}$, where $I_k \subset [n]$.

\textbf{Output}: 1 (accept) or 0 (reject).
 \begin{algorithmic}[1]
	\State $(\edges,\binID)\gets\prebinfp(\vecx,B)$
	\For{each tree $k\in[m]$}
	\State \textbf{parse} $(\vecf_k,\vectt_k,\vecw_k)\gets T_k$
		\For{each data point $i\in I_k$}
		\State compute the reached leaf $l_{k,i}\in[N]$ by evaluating $T_k$ on $\vecx_i$
		\EndFor
	\State $\vecl_k \gets (l_{k,i})_{i\in I_k}$
	\State \textbf{assert} $\validateinference((\vecx_i)_{i\in I_k},\vecf_k,\vectt_k,\vecl_k)=1$
	\State $H_k\gets\inithistslabel((y_i)_{i\in I_k},\binID,\vecl_k)$
	\State \textbf{assert} $\validateleafweightslabel(H_k,\vecw_k)=1$
	\State \textbf{assert} $\validatesplitsgini(H_k,\vecf_k,\vectt_k,\edges)=1$
	\EndFor
	\State \Return{1 (accept)}
 \end{algorithmic}
\end{procedure}

\begin{procedure}[label={prot:init-hists-label}]{$\inithistslabel$}
\textbf{Input:} Labels $\vecy=(y_i)_{i=1}^n$, $\binID$, leaf indices $\vecl=(l_i)_{i=1}^n$\\
\textbf{Output:} Histograms $H$
\begin{algorithmic}[1]
\State \textbf{for} $f\in[d],\,\ell\in[2N-1],\,b\in[B]$, initialize $H[f][\ell][b]\gets 0$.
\For{$i=1$ \textbf{to} $n$}
  \For{$f=1$ \textbf{to} $d$}
    \State $b\gets \binID[i][f]$;\; $\ell\gets l_i +(N-1)$
    \State $H[f][\ell][b][y_i] \gets H[f][\ell][b][y_i] + 1$
  \EndFor
\EndFor
\For{$\ell=N-1$ \textbf{down to} $1$} \Comment{propagate up to root}
  \For{$f=1$ \textbf{to} $d$}
    \For{$b=1$ \textbf{to} $B$}
      \State $H[f][\ell][b][0] \gets H[f][2\ell][b][0] + H[f][2\ell+1][b][0]$
      \State $H[f][\ell][b][1] \gets H[f][2\ell][b][1] + H[f][2\ell+1][b][1]$
    \EndFor
  \EndFor
\EndFor
\State \Return $H$
\end{algorithmic}
\end{procedure}

\begin{procedure}[label={prot:validate-leaf-weights-label}]{$\validateleafweightslabel$}
\textbf{Input:} Leaf histograms $H$, weights $\vecw=(w_l)_{l=1}^N$\\
\textbf{Output:} Accept/Reject
\begin{algorithmic}[1]
\For{$l=1$ \textbf{to} $N$}
  \State $\ell \gets l + (N-1)$
  \State $H_0^{\mathrm{leaf}} \gets \sum_{b=1}^B H[1][\ell][b][0]$ \Comment{use any feature $f$, as the sum is feature-independent}
  \State $H_1^{\mathrm{leaf}} \gets \sum_{b=1}^B H[1][\ell][b][1]$ \Comment{use any feature $f$, as the sum is feature-independent}
  \State $w'_l \gets \dfrac{H_1^{\mathrm{leaf}}}{H_0^{\mathrm{leaf}}+H_1^{\mathrm{leaf}}}$ 
  \State \textbf{assert} $w_{l} = w'_l$
\EndFor
\end{algorithmic}
\end{procedure}

\begin{procedure}[label={prot:validate-splits-gini}]{$\validatesplitsgini$}
\textbf{Input:} Node histograms $H$, resulting splits $(\vecf,\vectt)=(f_\ell,t_\ell)_{\ell=1}^{N-1}$, and $\edges$\\
\textbf{Output:} Accept/Reject
\begin{algorithmic}[1]
\For{$\ell=1$ \textbf{to} $N-1$}
  \For{$f=1$ \textbf{to} $d$}
    \State $H_0\gets\sum_{b=1}^B H[f][\ell][b][0]$;\quad $H_1\gets\sum_{b=1}^B H[f][\ell][b][1]$
    \State $H_{0,L}\gets 0,\,H_{1,L}\gets 0$
    \For{$b=1$ \textbf{to} $B$}
      \State $H_{0,L}\gets H_{0,L} + H[f][\ell][b][0]$;\quad $H_{1,L}\gets H_{1,L} + H[f][\ell][b][1]$
      \State $H_{0,R}\gets H_0 - H_{0,L}$;\quad $H_{1,R}\gets H_1 - H_{1,L}$
      \State $H_L \gets H_{0,L} + H_{1,L}$;\quad $H_R \gets H_{0,R} + H_{1,R}$
      \State $\gain[f][b] \gets \left(1 - \left(\frac{H_0}{H_0 + H_1}\right)^2 - \left(\frac{H_1}{H_0 + H_1}\right)^2 \right)  - \frac{H_{L}}{H_L+H_R} \cdot \left(1 - \left(\frac{H_{0,L}}{H_L}\right)^2 - \left(\frac{H_{1,L}}{H_L}\right)^2\right) - \frac{H_{R}}{H_L+H_R} \cdot \left(1 - \left(\frac{H_{0,R}}{H_R}\right)^2 - \left(\frac{H_{1,R}}{H_R}\right)^2\right)$
    \EndFor
  \EndFor
  \State Retrieve $b_\ell$ s.t.\ $t_\ell = \edges[f_\ell][b_\ell]$
  \State \emph{(Argmax check)} \textbf{assert} $\gain[f_\ell][b_\ell] \ge \gain[f][b]$ for all $f\in[d],b\in[B]$.
\EndFor
\end{algorithmic}
\end{procedure}
	\section{Additional Experiments}
\label{app:exp}
\subsection{Fixed-Point XGBoost Experiments}
\label{app:exptab}

This section provides the plaintext fixed-point XGBoost training versus the standard-precision floating point training. Details are in Tables~\ref{tab:credit},\ref{tab:covertype100k} \ref{tab:bc}, \ref{tab:covertype50k},\ref{tab:adult}.

\begin{table*}[t]
  \centering
  \footnotesize
    \caption{Credit Card Default ($n=30001$, $d=23$): Accuracy and runtime of our fixed-point algorithm vs XGBoost.}
  \begin{tabular}{@{}rrrrrrr@{}}
    \toprule
    \textbf{Depth} & \textbf{Trees} & \textbf{FixedAcc} & \textbf{XGBAcc} & $\lvert \Delta \rvert$ & $T_{\text{fixed}}$ (s) & $T_{\text{xgb}}$ (s) \\
    \midrule
  4 & 50 & 0.8204 & 0.8208 & 0.0004 & 2.70 & 0.13  \\
  \midrule
  4 & 100 & 0.8200 & 0.8185 & 0.0016 & 5.37 & 0.19  \\
  \midrule
  5 & 50 & 0.8216 & 0.8191 & 0.0025 & 3.41 & 0.15  \\
  \midrule
  5 & 100 & 0.8191 & 0.8171 & 0.0020 & 6.83 & 0.22 \\
    \bottomrule
  \end{tabular}
  \label{tab:credit}
\end{table*}
\begin{table*}[h]
  \centering
    \footnotesize
    \caption{Breast Cancer ($n=569$, $d=30$): Accuracy and runtime of our fixed-point algorithm vs XGBoost.}
  \begin{tabular}{@{}rrrrrrr@{}}
    \toprule
    \textbf{Depth} & \textbf{Trees} & \textbf{FixedAcc} & \textbf{XGBAcc} & $\lvert \Delta \rvert$ & $T_{\text{fixed}}$ (s) & $T_{\text{xgb}}$ (s) \\
    \midrule
  4 & 50 & 0.9737 & 0.9766 & 0.0029 & 0.48 & 0.02 \\
  \midrule
  4 & 100 & 0.9708 & 0.9766 & 0.0058 & 0.66 & 0.02 \\
  \midrule
  5 & 50 & 0.9708 & 0.9766 & 0.0058 & 0.56 & 0.02 \\
  \midrule
  5 & 100 & 0.9708 & 0.9795 & 0.0088 & 0.74 & 0.02 \\
    \bottomrule
  \end{tabular}

  \label{tab:bc}
\end{table*}

\begin{table*}[h]
  \centering
    \footnotesize
    \caption{Covertype - 50k ($n=50000$, $d=54$): Accuracy and runtime of our fixed-point algorithm vs XGBoost.}
  \begin{tabular}{@{}rrrrrrr@{}}
    \toprule
    \textbf{Depth} & \textbf{Trees} & \textbf{FixedAcc} & \textbf{XGBAcc} & $\lvert \Delta \rvert$ & $T_{\text{fixed}}$ (s) & $T_{\text{xgb}}$ (s) \\
    \midrule
  4 & 50 & 0.8186 & 0.8217 & 0.0021 & 7.26 & 0.12 \\
  \midrule
  4 & 100 & 0.8213 & 0.8235 & 0.0022 & 14.56 & 0.35 \\
  \midrule
  5 & 50 & 0.8203 & 0.8220 & 0.0017 & 9.32 & 0.35 \\
  \midrule
  5 & 100 & 0.8635 & 0.8610 & 0.0025 & 18.47 & 0.52 \\
    \bottomrule
  \end{tabular}
  \label{tab:covertype50k}
\end{table*}
\begin{table*}[h]
  \centering
  \footnotesize
    \caption{Covertype - 100k ($n=100000$, $d=54$): Accuracy and runtime of our fixed-point algorithm vs XGBoost.}
  \begin{tabular}{@{}rrrrrrr@{}}
    \toprule
    \textbf{Depth} & \textbf{Trees} & \textbf{FixedAcc} & \textbf{XGBAcc} & $\lvert \Delta \rvert$ & $T_{\text{fixed}}$ (s) & $T_{\text{xgb}}$ (s)  \\
    \midrule
    4 & 50  & 0.8216 & 0.8216 & 0.0000 & 13.64 & 0.47  \\
    \midrule
    4 & 100 & 0.8415 & 0.8410 & 0.0004 & 27.05 & 0.73  \\
    \midrule
    5 & 50  & 0.8382 & 0.8414 & 0.0032 & 16.52 & 0.50  \\
    \midrule
    5 & 100 & 0.8653 & 0.8633 & 0.0020 & 33.11 & 0.77  \\
    \bottomrule
  \end{tabular}

  \label{tab:covertype100k}
\end{table*}
\begin{table*}[h]
  \centering
    \footnotesize
    \caption{Adult ($n=45222$, $d=104$): Accuracy and runtime of our fixed-point algorithm vs XGBoost.}
  \begin{tabular}{@{}rrrrrrr@{}}
    \toprule
    \textbf{Depth} & \textbf{Trees} & \textbf{FixedAcc} & \textbf{XGBAcc} & $\lvert \Delta \rvert$ & $T_{\text{fixed}}$ (s) & $T_{\text{xgb}}$ (s) \\
    \midrule
    4 & 50  & 0.8633 & 0.8711 & 0.0078 & 11.48 & 0.31 \\
    \midrule
    4 & 100 & 0.8645 & 0.8719 & 0.0074 & 22.75 & 0.45 \\
    \midrule
    5 & 50  & 0.8617 & 0.8708 & 0.0091 & 15.09 & 0.32 \\
    \midrule
    5 & 100 & 0.8615 & 0.8711 & 0.0096 & 30.20 & 0.49 \\
    \bottomrule
  \end{tabular}

  \label{tab:adult}
\end{table*}

\subsection{Discussion for Fixed-Point XGBoost Experimental Results}
\label{app:expxgboost}


\noindent\textbf{Performance.}
As expected, our fixed-point implementation is slower than XGBoost, which is a highly optimized C++ library leveraging vectorization, cache tuning, and parallelism. However, once we level the playing field a bit by pinning both systems to a single thread, the slowdown becomes very manageable. On the Credit Default dataset, slowdowns range between $20\times$ and $30\times$, while on the larger Covertype-100k dataset they range from $29\times$ to $42\times$. With parallelism neutralized, the residual slowdown is consistent with constant-factor overheads of a Python/NumPy prototype, which are plausibly reducible in a lower-level implementation. A rewrite in a low level language with all the relevant optimizations should substantially reduce constant factors and, with the already-demonstrated parity in accuracy, plausibly bring slowdowns into the low double digits or better. In any case, although the inefficiency is important to note, the absolute running time which is around 30 seconds even on the larger datasets is still very practical.

\noindent\textbf{Small scale (Breast Cancer) as a counterpoint.}
On the very small Breast Cancer dataset (Table~\ref{tab:bc}), the slowdown is more pronounced in relative terms, averaging around $32\times$. This reflects the fact that on tiny workloads, the fixed costs of our Python-level and fixed-point implementation dominate, while optimized libraries like XGBoost maintain nearly constant runtime overhead.


\subsection{Multiclass Classification}\label{sec:multi-class}
\paragraph{Extension to multiclass classification.}
Our implementation and evaluation focus on binary classification, which already captures many prediction tasks such as fraud detection, credit-risk prediction, and mortality prediction. Nevertheless, the ZKBoost framework is not intrinsically limited to binary classification. The standard multiclass extension of gradient boosting maintains, for each sample \(x_i\), a vector of class logits
$$z_i = (z_{i,1},\ldots,z_{i,C})$$
rather than a single scalar logit, where $C$ is the number of classes. 
At each boosting round, the learner trains one tree per class, and the prediction probabilities are obtained by applying the softmax function
$$p_{i,c} = \frac{\exp(z_{i,c})}{\sum_{c'=1}^C \exp(z_{i,c'})}.$$

The gradients and Hessians used for split selection and leaf-weight computation then become class-indexed quantities, for example
$$g_{i,c} = p_{i,c} - \mathbf{1}[y_i=c]$$
and
    $$h_{i,c} = p_{i,c}(1-p_{i,c})$$, 
under the usual diagonal approximation used in multiclass tree boosting.

The certification approach of ZKBoost can extend to this setting. Instead of validating a single tree $T_k$ per boosting round, the certificate validates a collection $\{T_{k,c}\}_{c=1}^C$ of class-specific trees. The intermediate-score check is lifted from
scalar logits to class-indexed logit vectors, while the tree-local validation procedures for inference, histogram reconstruction, split optimality, pruning, and leaf weights are applied independently to each class-specific tree. Thus, the main structural contribution of $\certxgb$, which is separating inter-tree score consistency from tree-local bottom-up validation, applies unchanged. The proof cost scales linearly with the number of classes $C$, up to the additional cost of proving the softmax normalization.

A full multiclass implementation requires substantial engineering changes and is left for future work. However, preliminary measurements suggest that the additional softmax cost is not the dominant bottleneck. Using lookup-table techniques for exponentiation \cite{hao2024scalable} together with our division proof, a ZK proof of softmax evaluation costs approximately $50\,\mu s$ per data point per boosting round in our prototype, adding about $5\%$ overhead on the Covertype benchmark. 
We also implemented a fixed-point one-vs-rest multiclass variant as a sanity check for the numerical behavior of fixed-point tree boosting, and observed accuracy
comparable to floating-point XGBoost (see Table~\ref{tab:covertype50k-multiclass}). 
These results indicate that the main remaining challenge is more of an implementation effort rather than a conceptual limitation of ZKBoost.

\begin{table*}[t]
  \centering
  \footnotesize
  \caption{Covertype-50k ($n=50000$, $d=54$): Accuracy of our fixed-point one-vs-rest multiclass implementation vs XGBoost.}
  \begin{tabular}{@{}rrrr@{}}
    \toprule
    \textbf{Depth} & \textbf{Trees} & \textbf{Fixed-OvR} & \textbf{XGB} \\
    \midrule
    4 & 50  & 0.7845 & 0.7721 \\
    \midrule
    5 & 100 & 0.8338 & 0.8254 \\
    \midrule
    6 & 200 & 0.8759 & 0.8678 \\
    \midrule
    8 & 500 & 0.8891 & 0.8956 \\
    \bottomrule
  \end{tabular}
  \label{tab:covertype50k-multiclass}
\end{table*}

\subsection{ZKP Benchmarks}
\label{sec:appendix:perf:zkp}

The performance proving the comparison and histogram construction are shown in Table~\ref{tab:histogram}. 
The improvement for proof of comparison is due to the reduction in the prevention of wrap-around attack and the switching to the membership proof based on LogUp~\cite{DBLP:journals/iacr/Habock22a}.
Additionally, it benchmarks both the ZK-RAM construction suggested by~\cite{DBLP:conf/ccs/Pappas024} and our improvement with LogUp~\cite{DBLP:journals/iacr/Habock22a}. We use the VOLE-ZK-RAM for the former~\cite{yang2024two}. It shows that our histogram proof improves the naive implementation by $3-4\times$.

\begin{table}[t]
  \centering
  \footnotesize
    \caption{\small Left: Proof of comparison compared to the previous~\cite{hao2024scalable}. Right: Proof of histogram construction based on ZK-RAM and weighted Logup. Numbers in parentheses are the improvements.}
  \begin{tabular}{@{}ccc || ccc @{}}
    \toprule
   & LAN & WAN & & LAN & WAN \\\midrule
\cite{hao2024scalable} & 17.59 & 19.17 & w/ ZK-RAM~\cite{yang2024two} & 20.48 & 31.77  \\
Our Comparison & 3.29 ($5.3\times$) & 8.52 ($2.2\times$) & w/ Weighted LogUp~\cite{DBLP:journals/iacr/Habock22a} & 4.98 ($4\times$) & 10.35 ($3\times$) \\
    \bottomrule
  \end{tabular}
  \label{tab:histogram}
\end{table}
    \clearpage

 \end{document}